\documentclass[aps,prl,notitlepage,reprint,longbibliography,superscriptaddress]{revtex4-2}

\usepackage{times}
\usepackage{amsmath}
\usepackage{graphicx}
\usepackage{verbatim}
\usepackage{color}
\usepackage{caption}
\usepackage{subcaption}
\usepackage[export]{adjustbox}
\usepackage[pageanchor=true,
            plainpages=false,
            pdfpagelabels,
            bookmarks,
            bookmarksnumbered,
            colorlinks=true,
            allcolors=blue,
            pdfstartview=FitH]{hyperref}

\newtheorem{theorem}{Theorem}
\newtheorem{proof}{Proof}

\begin{document}

\title{Realizations of Isostatic Material Frameworks}
\author{Mahdi Sadjadi}
\affiliation{Department of Physics,
Arizona State University, Tempe, AZ 85287-1604}
\author{Varda F. Hagh}
\altaffiliation[Current address: ]
{James Franck Institute, University of Chicago, Chicago, IL 60637, USA}
\affiliation{Department of Physics,
Arizona State University, Tempe, AZ 85287-1604}
\author{Mingyu Kang}
\affiliation{CISE department, University of Florida, Gainesville, FL 32611-6120}
\author{Meera Sitharam}
\affiliation{CISE department, University of Florida, Gainesville, FL 32611-6120}
\author{Robert Connelly}
\affiliation{Department of Mathematics,
Cornell University, Ithaca, NY 14853}
\author{Steven J. Gortler}
\affiliation{School of Engineering and Applied Sciences, Harvard University, MA}
\author{Louis Theran}
\affiliation{School of Mathematics and Statistics, University of St Andrews, St Andrews KY16 9SS, Scotland}
\author{Miranda Holmes-Cerfon}
\affiliation{Courant Institute of Mathematical Sciences, New York University, NY}
\author{M.F. Thorpe}
\email[Corresponding author: ]{mft@asu.edu}
\affiliation{Department of Physics,
Arizona State University, Tempe, AZ 85287-1604}

\begin{abstract}
This paper studies the set of equivalent realizations of isostatic frameworks in two dimensions, and algorithms for finding  all such realizations.
We show that an isostatic framework has an even number of equivalent realizations that preserve edge lengths and connectivity.
We enumerate the complete set of  equivalent realizations for a toy framework with pinned boundary in two dimensions and study the impact of boundary length
 on the emergence of these realizations. To ameliorate the computational
complexity of finding a solution to a large multivariate quadratic system corresponding to the constraints;
 alternative methods - based on
constraint reduction and distance-based covering map or Cayley parameterization of the search space - are presented.
The application of these methods is studied on atomic clusters, a model two-dimensional glasses and jamming.
\end{abstract}
\maketitle


\section{Introduction}
A wide range of materials properties can be understood by modelling them
as mass-spring networks, or \emph{graphs with constrained edge-lengths}, where sites, or \emph{vertices}, are interacting via harmonic springs or \emph{edge-length constraints}. Examples include:
auxetic phases of matter~\cite{hagh2018auxetic} and mechanical metamaterials~\cite{bertoldi2017flexible,baardink2017localizing}. This network
representation contains topological and geometrical information. The topology of
a network determines how sites are connected while its geometry
determines the position of sites and in turn other geometrical properties
such as bond lengths and angles. Both geometrical and topological properties
of networks are crucial to control its response to mechanical
deformations which determines the rigidity of that structure~\cite{thorpe1983continuous}.
It is, therefore, not surprising that much research has been dedicated to tuning
materials properties by modifying the connectivity and geometry of networks.
Within the context of solid state physics, most studies have been focused on the topological design of networks in which
bonds are arranged such that the network response is optimized for a
given mechanical force/load~\cite{goodrich2015principle,hagh2018jamming}.

However, relatively less attention has been given to the geometrical \emph{realization} of
a network, i.e., the assignment of coordinates to
its sites in a given spatial dimension. In the study of geometric constraint
systems \cite{SitharamEtAl2018}, given a
graph $G$ with edge-length constraints, the realizations $p$ that satisfy those
constraints are called \emph{equivalent frameworks} $(G,p)$.
A given graph with constrained edge-lengths, can have many realizations. For example,
consider two triangles that share a common edge (4 vertices and five edges).
Here there are two realizations - one fully extended with no edges that cross
and the other folded about the common edge shared by the two triangles. This
gives a total of two realizations which is an example of the more general case
of an isostatic network having an even number of realizations that is shown and
extensively used in this paper.

The problem of finding network realizations has been applied to several physical
problems. The most well-known example is the so-called ``NMR problem'' where
pairwise distances between atoms are found using
nuclear magnetic resonance (NMR) spectroscopy~\cite{wuthrich1989protein} and
the three-dimensional protein conformation is inferred from the data~\cite{hendrickson1995molecule}.
Other examples include: survey and satellite imaging~\cite{killian1969einige},
localization of sensor networks~\cite{moore2004robust}, and conformation control
for allostery~\cite{rocks2017,kim2018conformational}.
Since the bond lengths are assigned to specific bonds, the problem is sometimes
referred to as
\textit{assigned distance problem} known to be NP-hard ~\cite{liberti2014euclidean}, while the
problem of finding realizations given only a list of bond lengths, that are not
assigned to specific bonds, is called the unassigned distance problem~\cite{duxbury2016unassigned}.

The network representation of a material is useful for both crystalline and
non-crystalline materials. The main difference is that crystalline structures
have only a single minimal energy conformation, while disordered
systems have a rough energy landscape with many
local minima. Hence, non-crystalline materials can attain various conformations
if such transitions are energetically accessible. In a large class,
the transition corresponds to a structural change in which atoms attain
new positions while the connectivity (the atoms they interact with) and the bond lengths remain unaltered, i.e., the conformations are in fact
equivalent frameworks.
For example, the anomalous properties of glasses such silica at low-temperature are
attributed to the two-level states (TLSs) in which the glass can tunnel between two conformations~\cite{Phillips1972,Anderson1972,Phillips1999}. These conformational
changes are believed to be localized, consistent with the thermal energy available
at about~$\sim1$ K where the anomalous specific heat is observed.
However, after half a century intensive research on TLSs, the geometrical
realizations, or equivalent frameworks, of these localized modes are still elusive.
The synthesis and
imaging of silica bilayers~\cite{Huang2012,lichtenstein2012atomic} in recent
years has reinvigorated open problems in physics of glasses by unveiling a structure which follows the continuous random model~\cite{rosenhain1927structure,zachariasen1932atomic}
and makes the actual coordinates of atoms available; albeit in two dimensions (2D). This newly available data on two dimensional glasses makes the interface between theory and experiment a lot easier; not least because visualization is so much easier in two dimensions.

The remainder of this paper is organized as follows. We first review the fundamental concepts in rigidity and present the theorem that states an isostatic network has an even number of realizations. Then we describe several methods to find realizations of an isostatic graph using a toy model,  using constraint reduction and  Cayley parameterization. Lastly, we apply these methods to a series of larger networks, generated computationally or experimentally, to find their realizations and discuss the physics of transition between such states.

\section{Mathematical Background}
\label{sec:background}
We aim to find all realizations of a network or graph with vertices connected by edges with given edge-lengths. A realization is the assignment of coordinates to vertices such that all edges satisfy their given lengths. A graph together with a realization is called a \emph{framework}. Frameworks that satisfy the same set of edge lengths are called \emph{equivalent.}
A realization is a solution to the set of edge length equations. Let $(x_i, y_i)$ be the coordinates of vertex $i$ in two-dimensions (2D). If vertices
$i$ and $j$ are connected through an edge with the length $s$, we can write:
\begin{equation}
  \label{eq:edgelens}
 (x_i-x_j)^2 + (y_i-y_j)^2 = s^2.
\end{equation}
Every edge in the graph has a corresponding \textit{edge length equation}. This is a geometric constraint problem that has been studied extensively from multiple perspectives from distance geometry, to algebraic geometry and automated geometry to structural or combinatorial rigidity and arises in a wide variety of applications. We refer a reader to a recent handbook for background, perspectives and recent work \cite{SitharamEtAl2018}.
An isostatic network has the minimum number of independent constraints or equations to make the graph  rigid, i.e. to ensure locally unique solutions \textit{generically} exist. It is this marginal state that separates overconstrained (more constraints than necessary for minimal rigidity) from underconstrained (fewer constraints than necessary for rigidity).
As mentioned earlier, in general, checking whether a real solution exists to such a  system of equations is known to be NP-hard.
This means that the source of the complexity  is not merely the number of solutions, which could be exponentially many in the size of the system.  In fact, even if there were just a single solution, finding it may take exponential time. Regardless of this complexity, it is possible to prove that \textit{A generic isostatic framework has an even number of realizations}.

This theorem is powerful as it suggests that glasses such as silica \textit{have to}
have more than one realization with the same topology (same set of edges and edge-lengths).
Now, note that the theorem guarantees the existence of such solutions, but
the question of their accessibility depends on the energy considerations. If the rigid bars between vertices are replaced by {\it{springs}} then there is an energy barrier between the various realizations, whose magnitude is relevant in physical process such as tunnelling.
In the next section, we justify this theorem using a toy model and will show
how various realizations of an isostatic framework can be found. To be
more precise~\cite{Hendrickson1992}:

\begin{theorem}
A finite generic isostatic framework is not globally rigid,
but has an even number of equivalent generic frameworks.
Each generic framework of the underlying graph is locally rigid.
(Equivalent generic networks have the same network topology and
bar lengths, and are infinitesimally rigid.)
\end{theorem}
\begin{proof}
This is essentially Theorem 5.9 from~\cite{Hendrickson1992}.
The evenness property
is not explicitly stated there, but is clear in the
proof. Evenness is explicitly stated in the proof of Theorem 1.14 from~\cite{Gortler}.
\end{proof}

However, this theorem does not provide a way to access  the solutions. Each realization
has exactly the same number of vertices and the same connectivity table,
and bond lengths, however the
embedding of the graph is different. These configurations are not related by rigid motions
such as translation and/or rotation.
An approach can be designed using the nature of an isostatic framework
which is on the verge of instability. The number of zero eigenvalues
of the dynamical matrix of an isostatic framework is exactly equal to the number
of trivial motions (or dimension of the null space). Any other motion has a finite cost in energy, if vertices are connected by springs, rather than bars which of course suppress any continuous deformations in an isostatic system.
But if a single constraint of an isostatic framework is removed, there are one fewer equations of the form Equation \ref{eq:edgelens}, so now the null space gains one extra dimension
moving along which has zero energy cost. In fact, it can be proven that the traversal
along this non-trivial eigenvector is continuous and leads to an even number of
realizations with the same length on the removed bar.
We state the above observations as a theorem.

\begin{theorem}
\label{singlecut}
If a single edge is removed from a finite generic isostatic  framework,
the resulting mechanism has a configuration space that is a closed,
continuous curve, on which there are an even number of  configurations
in which the removed edge returns to its original length.
\end{theorem}
\begin{proof}
See Theorems 5.8 and  5.9 in~\cite{Hendrickson1992}.
\end{proof}

\section{Trihex: A Toy Model}
Glasses like silica (SiO$_2$) and germenia (GeO$_2$) are considered as a network of
corner-sharing tetrahedra. Recently, silica bilayers have been
synthesized~\cite{Huang2012,lichtenstein2012atomic} which are effectively
a two-dimensional network of corner-sharing triangles~\cite{sadjadi17refining}.
These triangles are formed with oxygens at their corners. The network has rings of many
sizes but the mean ring size is six~\cite{Sadjadi2016}.
Therefore, we propose a hexagon as a toy model: Trihex (Fig.~\ref{fig:hexagon}).

\begin{figure}[t]
\centering
\includegraphics[width=8cm]{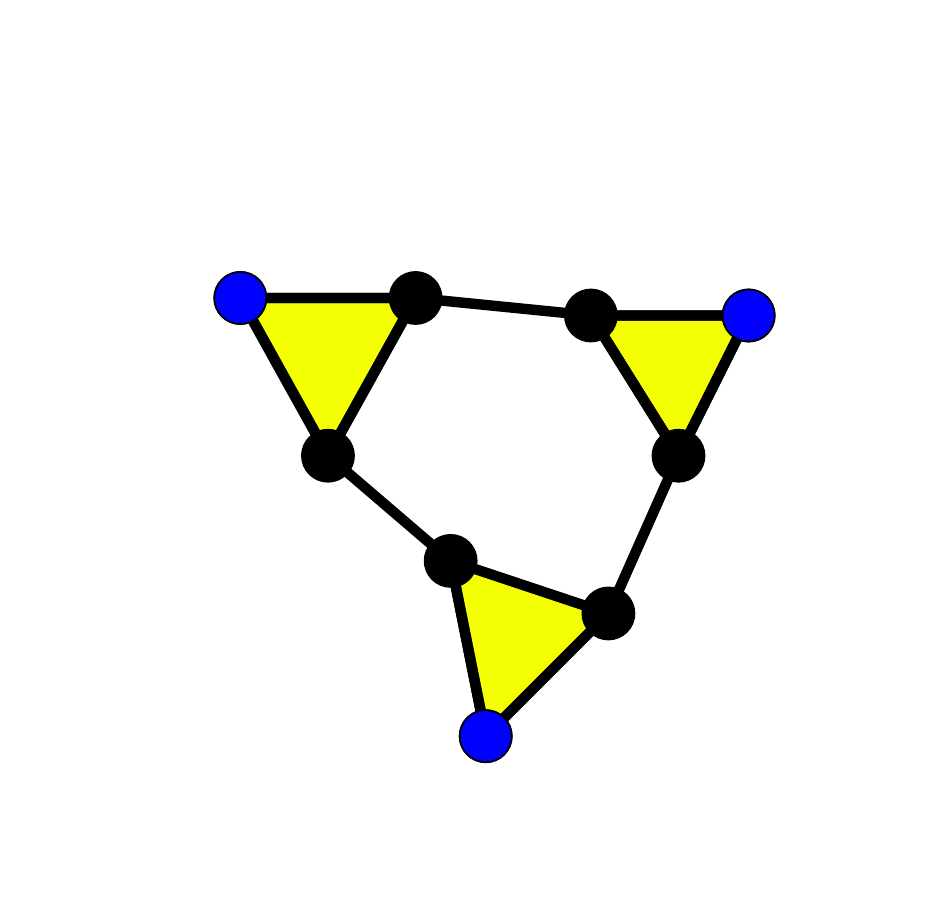}
\caption{\footnotesize The toy model, Trihex, formed by six corner-sharing triangles
proposed to find properties of realizations of an isostatic network. The other three non-pinned triangles are removed since each triangle is rigid.}
\label{fig:hexagon}
\end{figure}
It has $V=9$ vertices and $E=12$ edges. To render Trihex isostatic, anchored
boundary condition are used and the blue vertices are pinned
(immobilized)~\cite{theran2015anchored}. The pinned vertices are placed
generically (\textit{not on an equilateral triangle}) and all edges initially have almost, but not exactly, the same length. The other three vertices on the
surface are removed since they do not change the rigidity of the network
as each triangle is rigid.
 Since the set $V$  has six unpinned vertices, this gives a set of $2N = 2*6 = 12$
non-linear equations for Trihex to solve using three pinned vertices whose coordinates are fixed.
It is important to note that simple ruler and compass based  algorithms, that classify and find equivalent frameworks when the underlying graphs are so-called tree-decomposable or quadratically solvable,   \cite{fudos1997graph, owen1991algebraic, owen2007nonsolvability, joan-arinyo2004revisiting} do not directly apply to Trihex.

\section{The Single-cut algorithm}
 Theorem \ref{singlecut} can be directly written as a step-by-step \textit{single-cut}
algorithm~\cite{HC14}:
\begin{enumerate}
  \item Start from an isostatic network, {\it i.e.} a rigid network with
  no redundant edge. The number of trivial motions depend on the imposed boundary
  conditions. In a system with periodic boundary conditions, only rigid translations
  are allowed. For anchored boundary condition, no trivial motion exists, and there are exactly $2N$ equations of the form Equation \ref{eq:edgelens}.
  \item Remove an edge from the isostatic network, resulting in a system of  $2N-1$ equations and form its dynamical matrix.
  Find the eigenvectors corresponding to zero eigenvalues (the null space). Remove
  trivial motion eigenvectors to find the one internal degree of freedom (\emph{dof}).
  \item \textit{Eigenvector-following}: Once the non-trivial direction
  is identified, move all sites along that direction with a small step size. The smaller
  the step size, the smaller is the error in traversing the closed curve in configuration space, \textit{i.e.} the path
  that the system takes in high dimensional space. Note that this motion does not
  change the edge length of any other edge.
  Also use the dot product of the previous and current directions to make sure we
  only move forward in the configuration space.
  \item Compute the dynamical matrix at the new point and repeat the above process
  to traverse in the configuration space. Meanwhile monitor the distance between
  the two vertices that had their connecting edge removed. If we continuously move through
  this one-dimensional path, we eventually come back to the starting point. Once
  we are back to the initial point, the sum of distances from the center of mass (as a convenient metric)
  is plotted against the length of the cut edge, for each point along the path.
  This gives us a closed curve projected in 2D plane in which drawing a vertical line
  will identify the original framework and its equivalent frameworks in the configuration space.
\end{enumerate}
A Python implementation of this algorithm can be found in the RigidPy package~\cite{rigidpy}. The single cut algorithm is illustrated in the two parts of Figure~\ref{fig:circuits}.

\begin{figure}[t]
  \centering
  {\includegraphics[width=7.5cm]{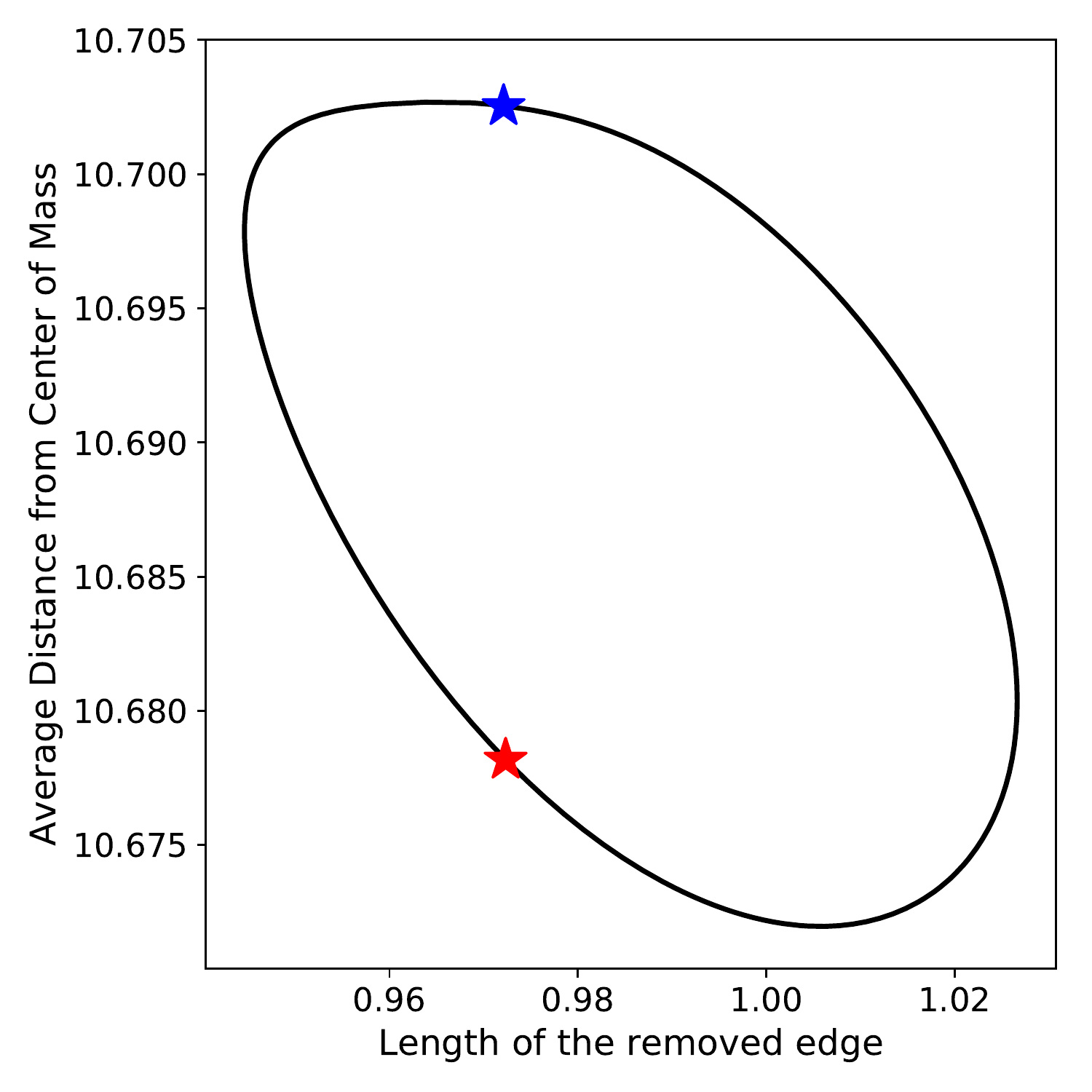}}
  {\includegraphics[width=7.5cm]{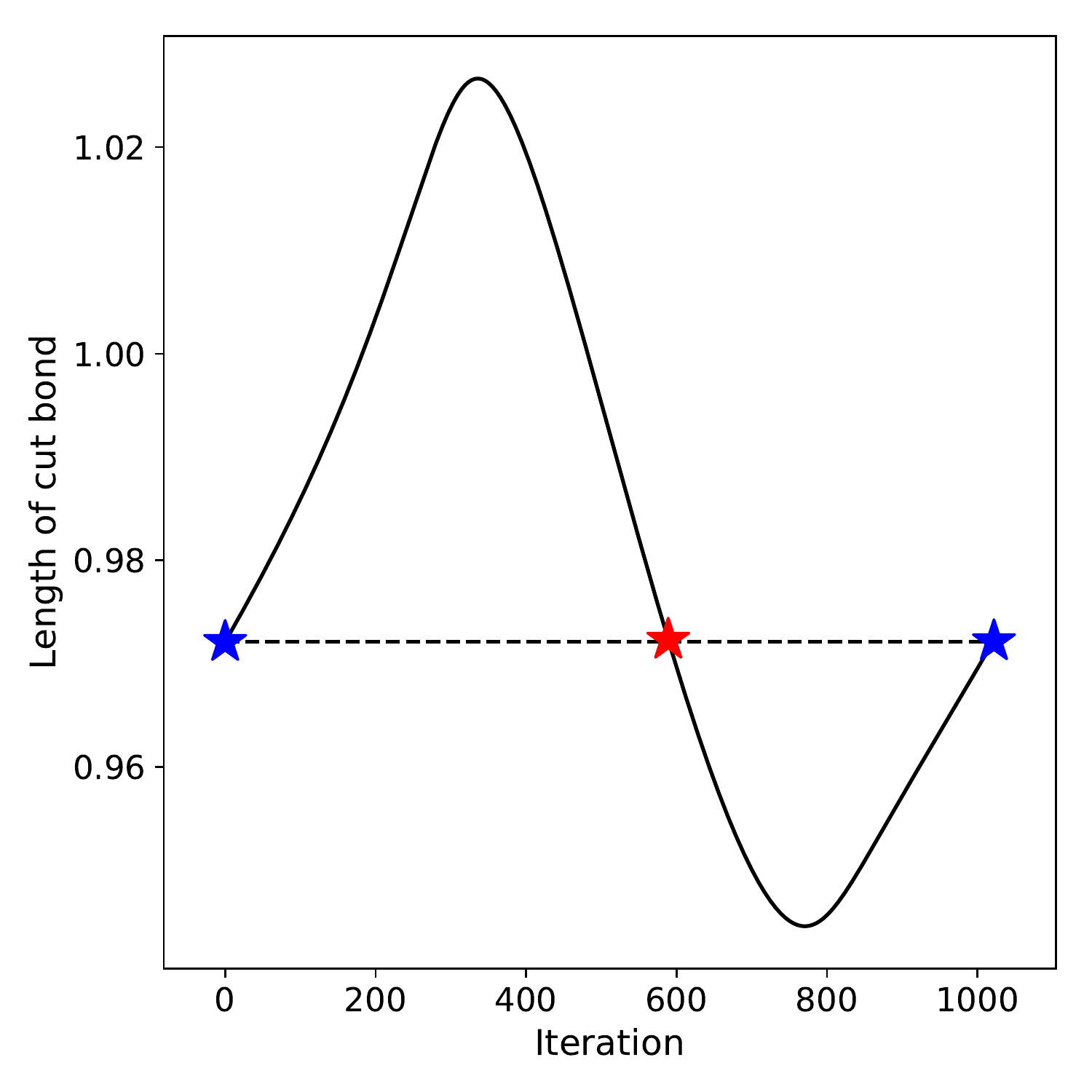}}
  \caption{\footnotesize (Top) A closed curve - whose points represent equivalent frameworks or configurations of a 1-dof framework - projected on the 2D plane; The vertical
  axis represents the average distance of all vertices from the center of mass.
  The horizontal axis shows the distance between two ends of the removed edge. The blue
  and the red asterisks denote the original and alternative realizations,  i.e. equivalent frameworks.
  Note that more than 2 (but an even number) of such equivalent frameworks can be found on a single closed curve (see Section on CayMos).
  A vertical line, drawn at the location of the original bond length, has two
  intersections with the  closed curve representing equivalent frameworks.
  (Bottom) The distance of two ends of the removed edge vs. iteration step
  by moving along the path. The dashed horizontal line represents the original length.
  The asterisks correspond to the ones on the left.}
  \label{fig:circuits}
\end{figure}



\subsection{Using the single-cut algorithm to find equivalent sphere packings}

We tested the ability of the single-cut algorithm to find equivalent configurations of frameworks using a large database of known framework embeddings which were obtained from rigid unit sphere packings of $N=12,13$ spheres \cite{HC14}. For $N=12$, the database in \cite{HC14} contains 11,980 distinct unit sphere packings, of which there are 23 pairs with the same adjacency matrix and hence the same underlying framework graph. We applied the single-cut algorithm to each of the 46 frameworks with multiple embeddings, breaking each single bond in turn. For all frameworks the algorithm found the other embedding, usually via several different single broken bonds.
For $N=13$, the database contains 98,529 unit sphere packings of which there are 474 pairs with the same adjacency matrix. We tested all the frameworks with multiple embeddings, and found 4 pairs of frameworks that couldn't reach their other embedding by the single-cut algorithm. Interestingly, an additional two pairs (4 frameworks) each led to new frameworks, that the algorithm in \cite{HC14} failed to find.
A pair of frameworks that cannot be converted to each other via the single-cut algorithm is shown in Figure~\ref{fig:clusters}. These atomic clusters are three dimensional and indeed Theorems 1 and 2 and the single cut algorithm apply in any dimension. The remaining examples in this paper are in two dimensions.

\begin{figure}[ht]
\centering
\includegraphics[width=0.9\linewidth]{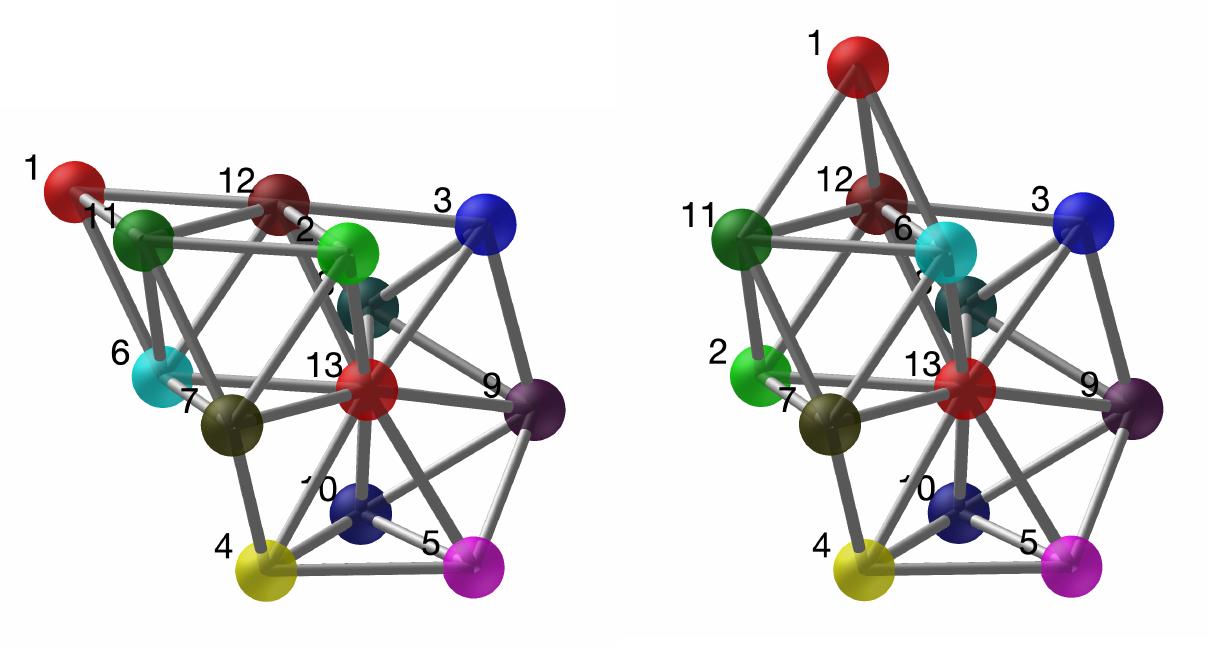}
\caption{A pair of frameworks with the same underlying graph but distinct spatial embeddings, that cannot be converted to each other via the single cut algorithm. The embeddings differ in the locations of vertices 2 (green), 6 (cyan), and 1 (red). By breaking a single bond and flexing, there is no way to interchange the positions of vertices 2 and 6; more flexibility is required to interconvert the embeddings.}
\label{fig:clusters}
\end{figure}

\subsection{Double cut}
In order to find other realizations in this system, we designed more complex
schemes for removing edges. For example, the single-bond cut can be modified
to cutting two bonds while another bond is added. We tried different ways in
connecting different sites but our analysis shows that these more complex schemes
do not find any new realizations. We do not find new solutions
in these more complex schemes, but even more complex schemes can be imagined that perhaps might do so.

\subsection{An Alternative to address Incompleteness of single cut}
It is helpful to look at the configurations of the Trihex where the anchored vertices are close to being a unit-length equilateral triangle and edge lengths are larger that $1/3$, and smaller that about $1/2$. Otherwise, either the configuration or framework does not exist, or it has several self-intersections which do not seem physically realistic.  As we will see, this allows for a great variation in the shapes of  equivalent frameworks.

First we consider the non-generic case when the anchored triangle is exactly equilateral, which leads to some flexible frameworks, although generic frameworks of the same graph are in fact isostatic. Subsequently, we will consider generic perturbations of an anchored triangle. In Figure \ref{fig:Clock+}, sample  equivalent frameworks are displayed in the form of the numbers on a clock, where each hour represents a particular configuration and there is a natural flex of frameworks from each hour to the next, forward or backward. The central hexagon in each of the flexible configurations have the property that opposite sides are parallel.  All the dark edges have unit length.  The two equivalent configurations in the center are rigid isostatic, and also infinitesimally rigid configurations that have $3$-fold rotational symmetry, which we call $3$-fold left and $3$-fold right.  They are not part of the flexible cycle of configurations on the outside clock, but they also represent realizations of the outside graph with the same edge lengths and a connection is shown for each of these to the particular odd hour configuration that share two of the three triangle configurations that are the same.  The flex is seen by starting from the $12$-o-clock configuration, and using that the opposite sides of the hexagon are parallel it maintains that property as it flexes.  This is proved below.  If you look at the path of $3$ edges from one vertex to the next, it is an example of a $4$-bar mechanism, and  in the clock flex, each of the $3$~$4$-bar mechanisms flex their full cycle, and when the edges of the pinned triangle are greater than the length of the inner edge length, the  $4$-bar mechanisms consist of one connected component.

\begin{figure}[b]
\centering
\includegraphics[width=8cm]{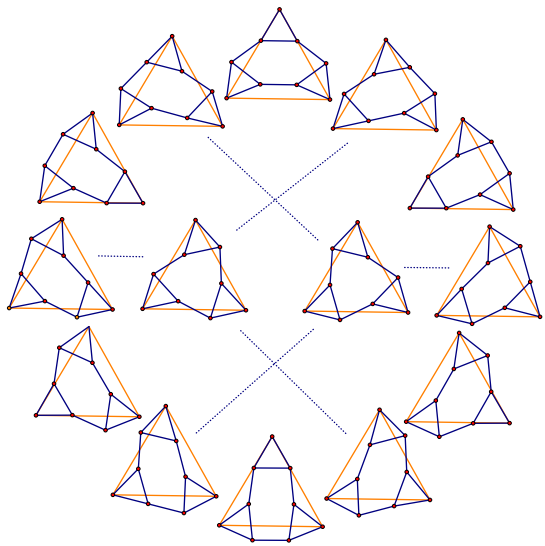}
\caption{\footnotesize The closed path of the flexing trihex for the case when the pinned vertices form and equilateral triangle shown in red.  The positions of the sample configuration are represented by an hour on a clock.  The central two $3$-fold symmetric configurations are rigid and not part of the flex.}
\label{fig:Clock+}
\end{figure}

To prove that the clock mechanism works see Figure \ref{fig:Vectors}.  We use the $1$-o-clock position as a sample.  The more general position is very similar.  We see that the vector sum $A+B+C=D$ represents the vectors of the $4$-bar mechanism, where $D$ is the corresponding side of the regular pinned triangle.  Let $R$ be the rotation to the right by $120^{\circ}$. Then $RC+RA+RB=R(C+A+B)=RD$ is the other side of the pinned triangle. Applying this to other side we can see how things close up.

\begin{figure}[b]
\centering
\includegraphics[width=6cm]{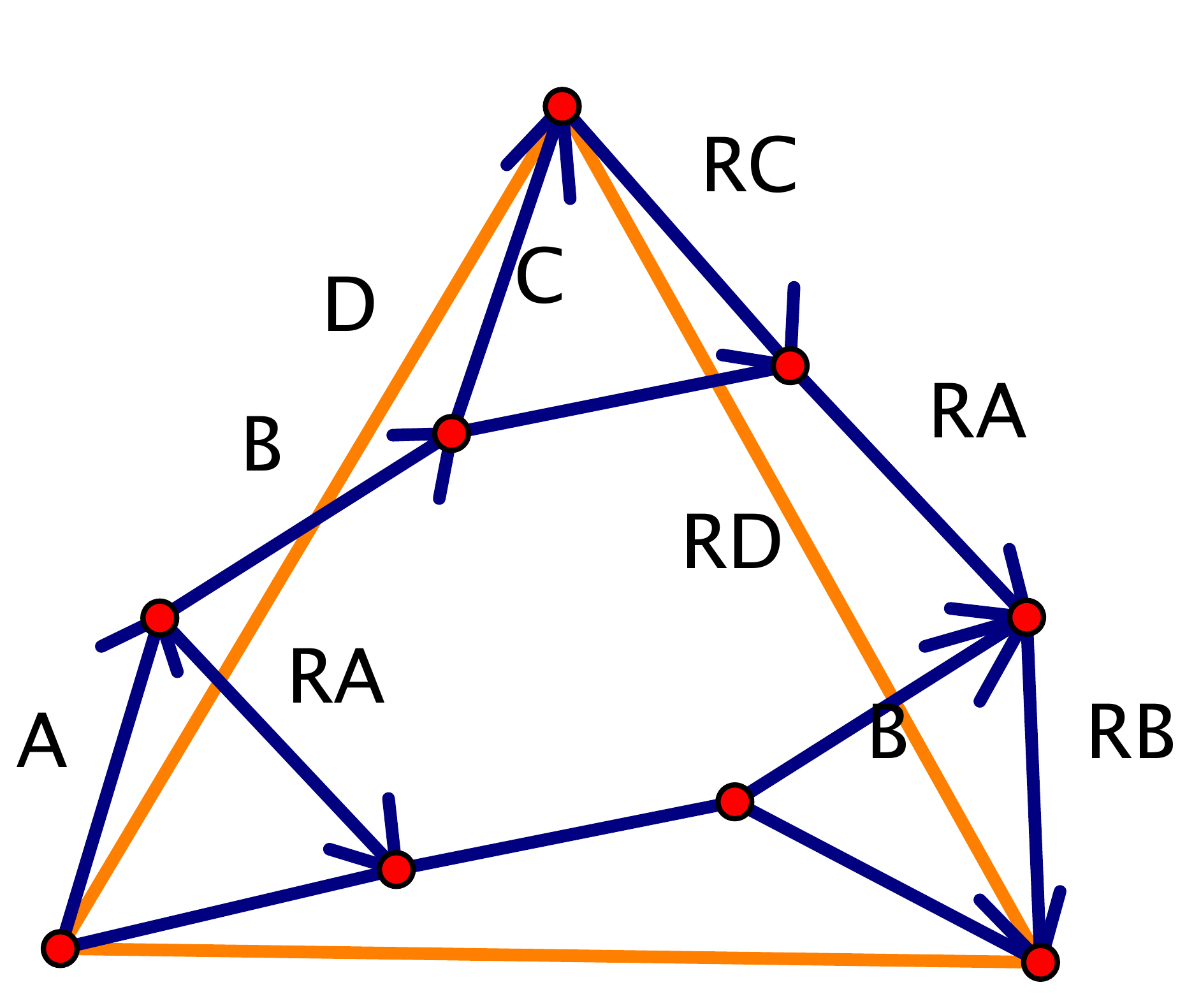}
\caption{\footnotesize This labels the directed edges of mechanism for the proof that the clock mechanism is a mechanism.}
\label{fig:Vectors}
\end{figure}

\begin{figure}[b]
\centering
\includegraphics[width=8cm]{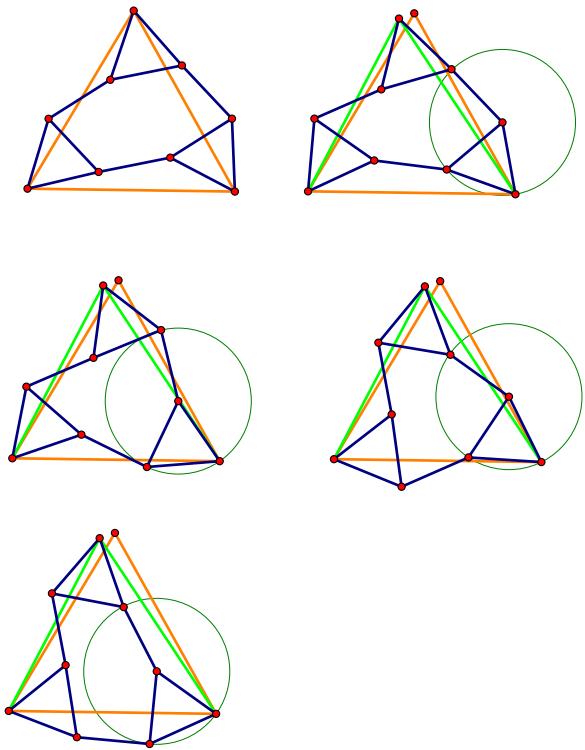}
\caption{\footnotesize Starting with the framework in the top left, which is near the $1$-o-clock flexible configuration, we perturb the starting trihexagon to get the framework on the upper right.  We then find $3$ others as shown.  The thin circle is to indicate that the last bar length is of unit length since the other bar lengths are of unit length by construction.}
\label{fig:flexings-1-o-clock}
\end{figure}

We next show how this applies when the pinned triangle is perturbed to a non-equilateral, generic triangle, while the other other unit bar lengths are fixed.  The two $3$-fold configurations are infinitesimally rigid, and the whole framework is isostatic with $9$ vertices and $15= 2\cdot 9 -3$ bars.  So it has only the $0$ equilibrium (self) stress.     On the other hand, for perturbations of configurations of the clock mechanism it is not so simple.  Indeed, there is an equilibrium stress that varies as the configuration flexes, and in a sense, the stress ``blocks'' some of the infinitesimal and actual motions.

For a perturbation of the pinned triangle, we may assume that  one of the edges of the corresponding equilateral triangle is the same length, and each of the other two edge lengths either increase or decrease some small amount.  If the stress on those lengths has the same sign as the displacement of the edge lengths, that motion is restricted.  But in any case, we can start with one of the configurations along the path of the clock mechanism, and then look for another realization with the same perturbed edge lengths of the pinned triangle.  Figure \ref{fig:flexings-1-o-clock} is an example of that process.  Then fixing the new green non-equilateral base triangle, we can then find two other realizations with unit bar lengths on the dark colored bars which are perturbations of the left and right $3$-fold.

When the upper left configuration is flexed $6$ hours, then we can approximate that as well to get an exact trihex with the same triangle base as in the upper right.  Thus we get $4$ configurations with the same base altogether.

The above discussion indicates that there are various other equivalent frameworks  than obtainable by removing an edge and flexing the resulting framework.  Since there is a finite mechanism nearby, or a critical configuration nearby, with generic equilateral triangle of boundary vertices, that can be used as a kind of guide path, walking around the clock to find distant realizations with the same bar lengths. The idea is to approximate the given framework with one with the nearest time on the clock.  For example the framework in Figure \ref{fig:hexagon} is roughly at $7$-o-clock, when the clock in Figure \ref{fig:Clock+} is rotated $180^{\circ}$ to match the position of the base triangle. In fact, that is not all.  It is possible to jump to both of the left or right folds rather that walk around the clock and use that as the approximation.

In the next section, we present a formal method that fleshes out these ideas.

\section{Indexing and finding equivalent frameworks using Cayley Parameters}

The methods mentioned above use two ideas. The first is  finding equivalent frameworks by removing one or more edges, and exploring the configuration space of the resulting mechanism with one or more degrees of freedom, to find other configurations that satisfy the original edge lengths of the removed edges, i.e., equivalent frameworks. The second is to classify or index the configurations based on their ``clock'' position around critical configurations.

These ideas are exploited in a formal method that has been used in a variety of scenarios from computer aided engineering design to molecular modeling.
Cayley parameterization (a type of \emph{covering map or projection} in the terminology of algebraic topology)  was introduced in \cite{SiGa:2010}, as a way of describing and computing configuration spaces of flexible frameworks. The key idea is to use lengths of selected non-edges as parameters or coordinates, to represent and traverse the configuration space that could be disconnected in the usual Cartesian coordinates (a \emph{branched covering space} in the terminology of algebraic topology). With sufficiently many judiciously chosen Cayley parameters or non-edges whose addition makes the framework minimally rigid or isostatic, we can  efficiently compute the finitely many possible Cartesian configurations  (inverse of the covering map) corresponding to each Cayley-parameterized configuration (an element of the \emph{base space} of the covering projection). A bijection between Cayley configurations and frameworks is achieved by adding enough Cayley parameters, i.e. enough non-edges so that the graph is globally rigid, i.e. has a unique realization generically, given edge lengths.

Here we describe two algorithms based on Cayley parameterization for finding all equivalent frameworks of a given framework, or graph with fixed edge-lengths.

The first is based on results in  \cite{sitharam2014beast, sitharam2011cayley1, sitharam2011cayley2, wang2015caymos}, which provided an analysis of a  common class of 1-dof mechanisms  that are obtained  by removing an edge from  the well-studied class of tree-decomposable  graphs, which include the so-called Henneberg-I graphs \cite{gao2009henneberg}.

This class of graphs provides a natural classification or indexing of equivalent frameworks based on relative orientations, chiralities, or \emph{flips} of certain triples of vertices. A flip vector of 1's and -1's distinguishes equivalent isostatic frameworks.
Two equivalent isostatic frameworks whose flip vectors differ in a single coordinate  are on ``opposite sides'' of a critical configuration of  a 1-dof framework where the triple of vertices is collinear. The isostatic framework minus a ``base'' edge yield the 1-dof framework, and
the Cayley configuration space of this 1-dof framework is parameterized by the length of the removed base edge, i.e.  base non-edge. This Cayley configuration space has a well-defined structure of intervals bounded by critical points. Each interval could  correspond to multiple connected component curves that are generically smooth and are homeomorphic to a circle. Certain pairs of flip vectors are guaranteed to belong to different components, while others could belong to the same component.
Moreover, once the unique flip vector is given, the  isostatic framework with the specified length of the base non-edge or Cayley parameter can be found easily using a simple ruler and compass construction.
The algorithm has been implemented as an opensource software CayMos.

Trihex becomes a 1-dof tree-decomposable graph after removing one edge $e$, with a Cayley configuration space parameterized by the length of a  \emph{different} base non-edge $f$, whose addition makes the graph tree-decomposable.
As the different connected component curves of the  configuration space are traced out for the different lengths of $f$,  multiple equivalent frameworks, indexed by  multiple flip vectors, that attain the required length for $e$ are found.  A webpage  \cite{caymostrihexwebpage} illustrates CayMos analysis of the Trihex for various ratios between the pinned boundary edge length and the bond length. The two component curves of a  Trihex with bond length  half the pinned boundary edge length  are shown in Figure \ref{fig:caymoscomponents}. Each component curve could have multiple equivalent frameworks with  different flip vectors, i.e. frameworks that attain the required distance of the dropped edge.

\begin{figure}
    \begin{subfigure}[t]{\linewidth}
    \centering
    \includegraphics[width=\linewidth, valign=t]{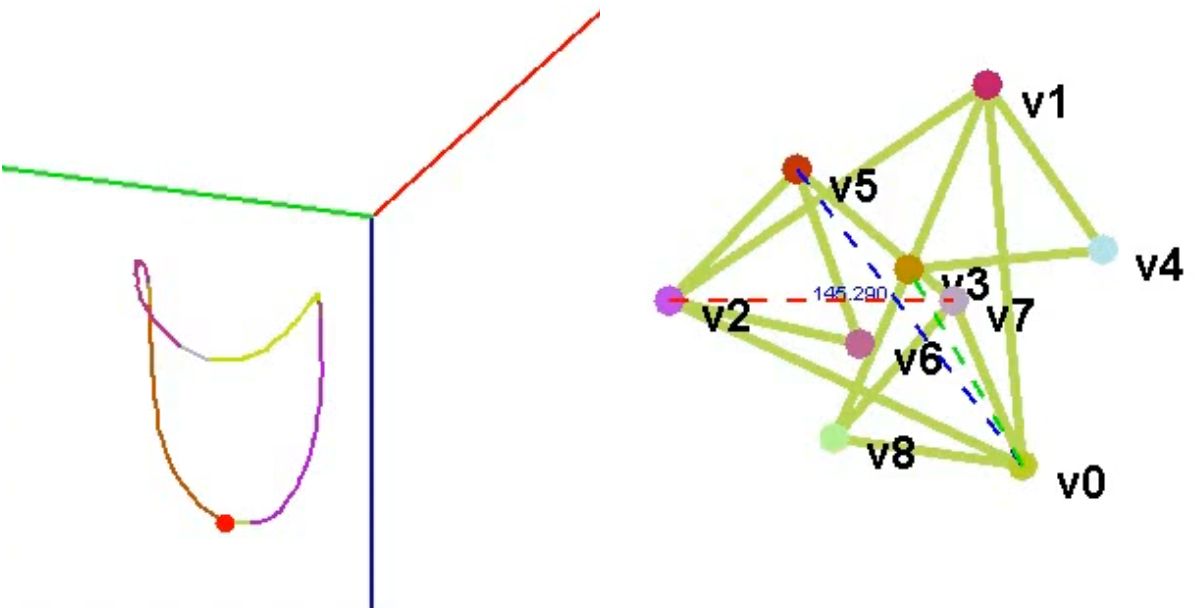}
  \caption{\footnotesize Component curve 1, contains all flip vectors with anticlockwise orientation of the triple $v0,v7,v8$.}
    \label{fig:component1}
    \end{subfigure}
    \vspace{2em}
    \begin{subfigure}[t]{\linewidth}
    \centering
    \includegraphics[width=\linewidth, valign=t]{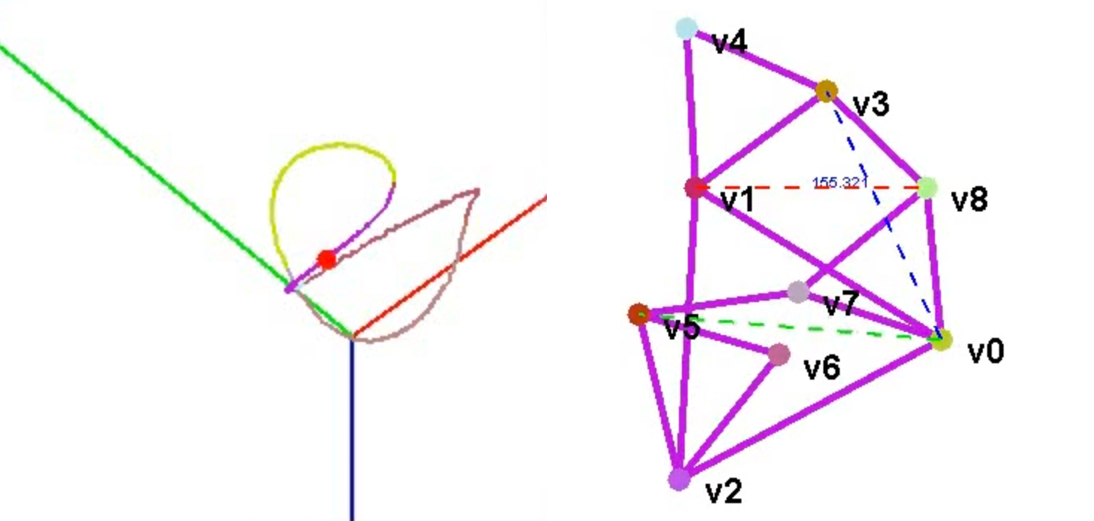}
 \caption{\footnotesize Componentcurve 2, contains all flip vectors with clockwise orientation of the triple $v0,v7,v8$}
    \label{fig:component2}
    \end{subfigure}
    \caption{\footnotesize Trihex minus one edge 1-dof tree decomposable graph's Cayley configuration space curves for bond length half of boundary edge length, solved by CayMos. The dropped edge $(v4,v6)$ attains its required length at multiple points on each curve, giving multiple equivalent frameworks. See text for description}
    \label{fig:caymoscomponents}
\end{figure}

The dropped edge $e$ is $(v4,v6)$.  Each component curve is a projection of a smooth simple curve (each point represents a unique configuration) living in 3D, parameterized by 3 Cayley parameters, dashed non-edges whose addition makes the graph globally rigid, i.e., generically has a unique realization given edge lengths. The  driving Cayley parameter or base non-edge is one of the 3 dashed. Each colored portion of the curve represents a different flip vector. Each component curve has multiple equivalent frameworks with different flip vectors, i.e. frameworks that attain the required distance of the dropped edge.

Figure \ref{fig:caymos-solution} gives equivalent frameworks for Trihex
found by CayMos, for two different ratios between the edge length of the pinned (nearly) equilateral boundary triangle and the (equal) length of the remaining edges.

\begin{figure}
    \begin{subfigure}[t]{0.45\linewidth}
    \centering
    \includegraphics[width=\linewidth, valign=t]{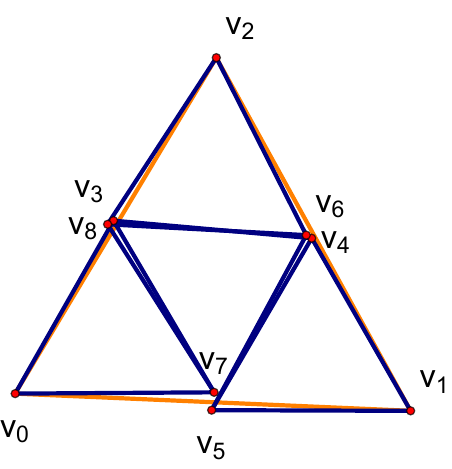}
    \caption{\footnotesize When the edge length is \(0.503\) and CayMos finds 22 realizations, see table.}
    \label{fig:trihex1}
    \end{subfigure}
    \vspace{2em}
    \begin{subfigure}[t]{0.45\linewidth}
    \centering
    \includegraphics[width=\linewidth, valign=t]{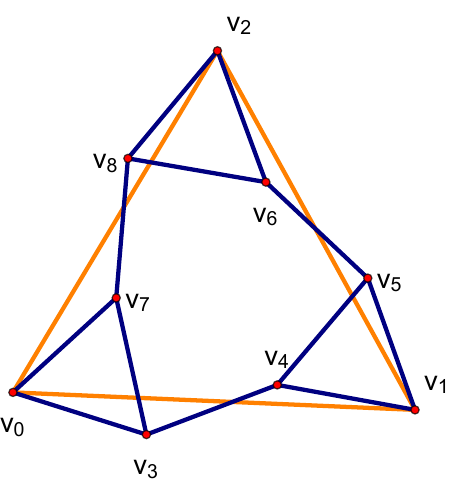}
    \caption{\footnotesize When the edge length \(0.348\) and CayMos finds 4 realizations, see table}
    \label{fig:trihex2}
    \end{subfigure}
    \caption{\footnotesize Trihex frameworks for two different boundary-to-edge-length ratios, solved by CayMos.}
    \label{fig:caymos-solution}
\end{figure}

\begin{table}
\label{tab:trihex1_sol}
\caption{\footnotesize 22 solutions found for trihex \ref{fig:trihex1}}
\begin{tabular}{c|c|c}
    \hline
    Flip & Base Edge & Solutions \\ \hline
    \(-\)& \( \left( v_2,v_5 \right) \) &
    \begin{tabular}{cc}
        \includegraphics[width=0.25\linewidth]{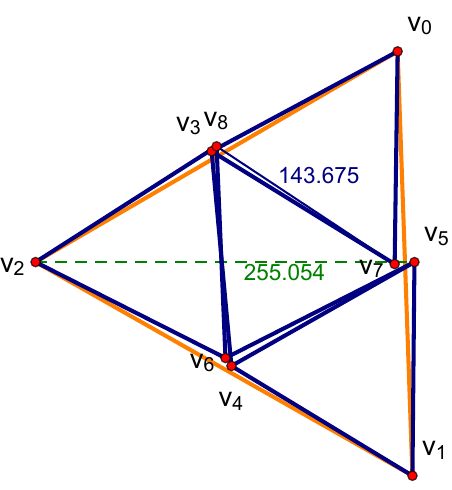} &
        \includegraphics[width=0.25\linewidth]{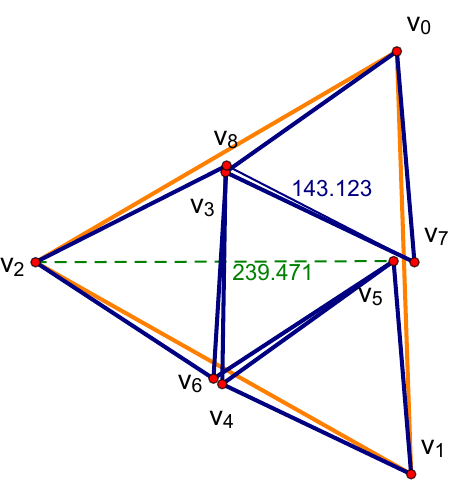}
    \end{tabular} \\ \hline
    \( 4 \)& \( \left( v_2,v_5 \right) \) &
    \begin{tabular}{cc}
        \includegraphics[width=0.18\linewidth]{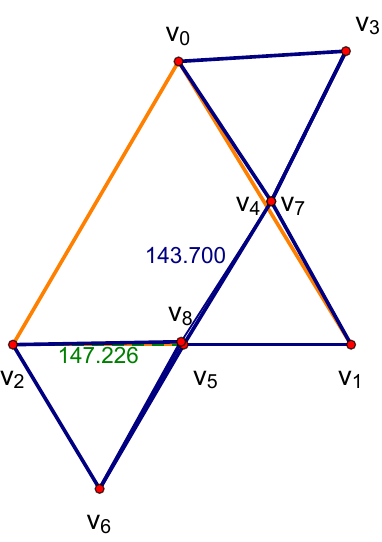} &
        \includegraphics[width=0.25\linewidth]{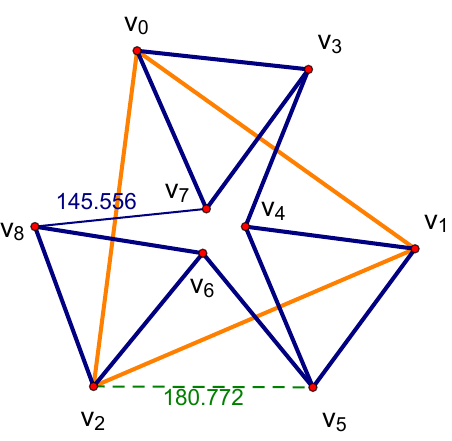}
    \end{tabular} \\ \hline
    \( 7 \) & \( \left( v_2,v_5 \right) \) &
    \begin{tabular}{cc}
        \includegraphics[width=0.22\linewidth]{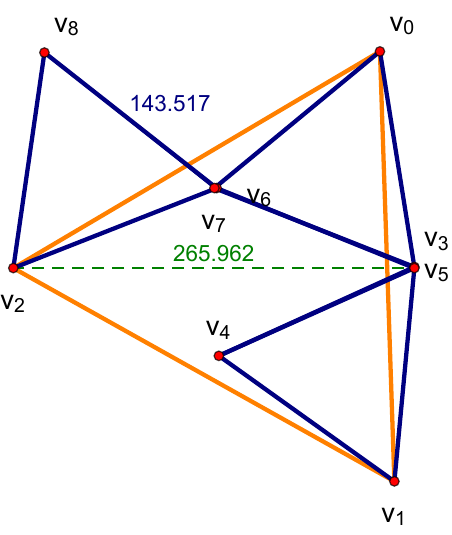} &
        \includegraphics[width=0.22\linewidth]{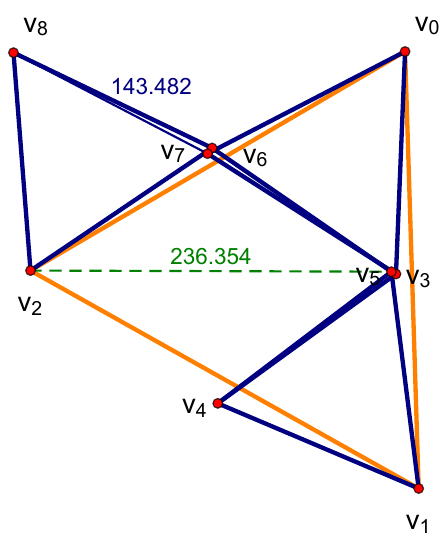}
    \end{tabular} \\ \hline
    \( 8 \) & \( \left( v_2,v_5 \right) \) &
    \begin{tabular}{cc}
        \includegraphics[width=0.25\linewidth]{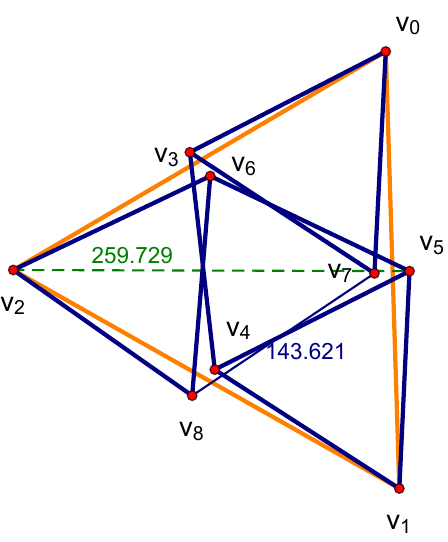} &
        \includegraphics[width=0.25\linewidth]{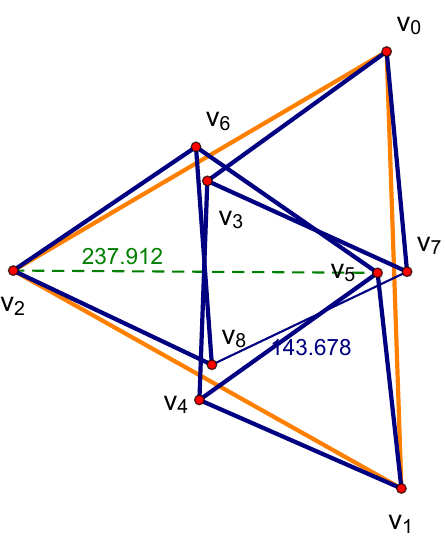}
    \end{tabular} \\ \hline
    \( 3,\,4 \) & \( \left( v_2,v_5 \right) \) &
    \begin{tabular}{cc}
        \includegraphics[width=0.25\linewidth]{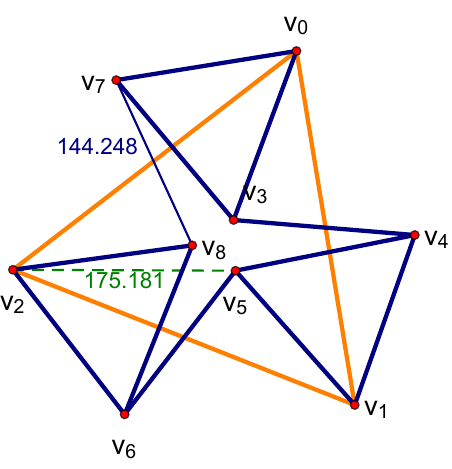} &
        \includegraphics[width=0.25\linewidth]{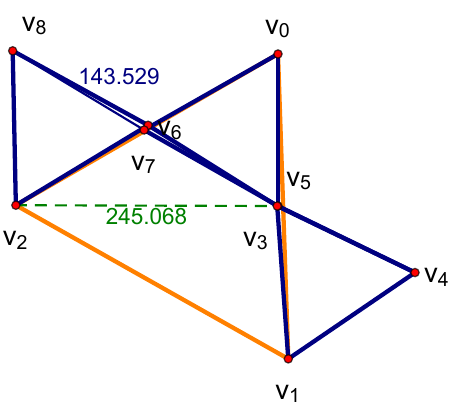} \\
        \includegraphics[width=0.25\linewidth]{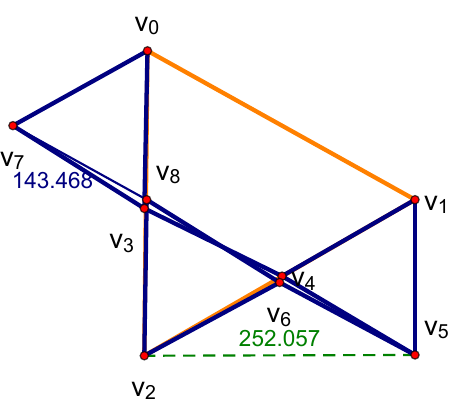} &
        \includegraphics[width=0.25\linewidth]{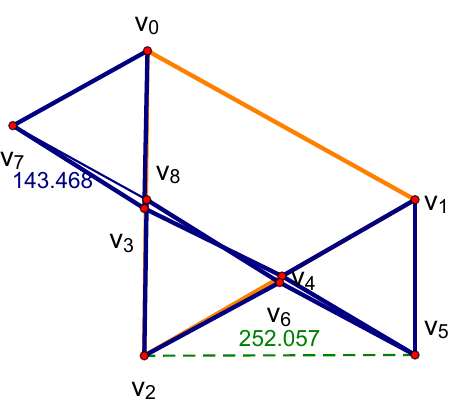}
    \end{tabular} \\ \hline
    \( 7,\,8 \) & \( \left( v_0,v_4 \right) \) &
    \begin{tabular}{cc}
       \includegraphics[width=0.25\linewidth]{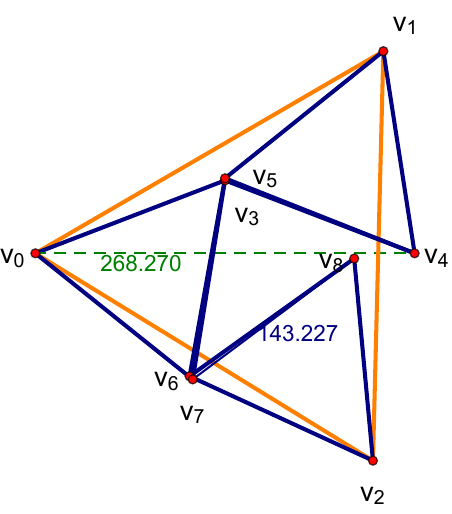} &
       \includegraphics[width=0.25\linewidth]{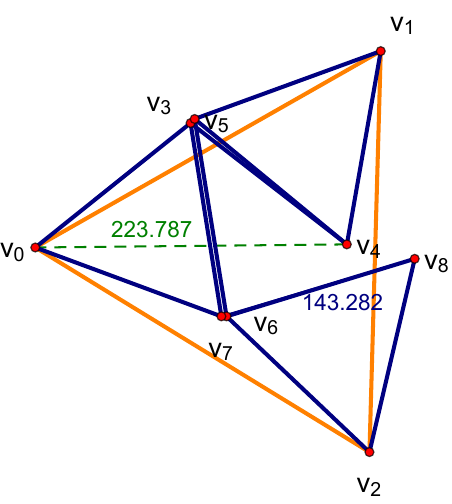}
    \end{tabular} \\ \hline
    \( 3,\,4,\,7 \) & \( \left( v_0,v_4 \right) \) &
    \begin{tabular}{ccc}
        \includegraphics[width=0.20\linewidth]{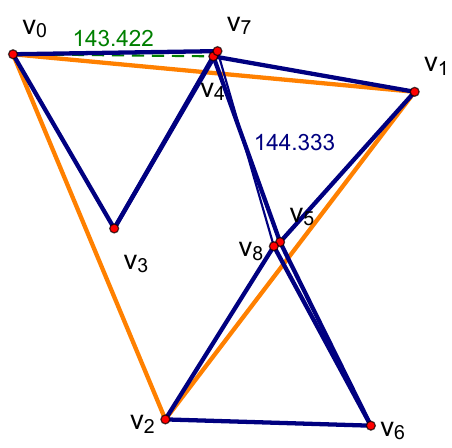} &
        \includegraphics[width=0.20\linewidth]{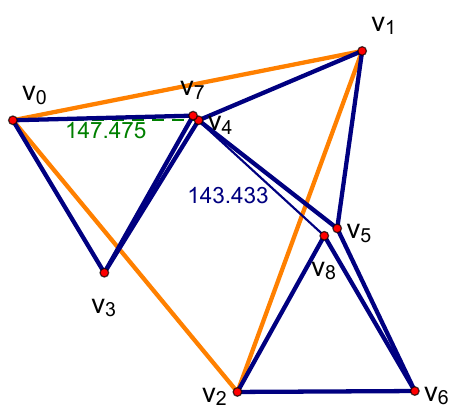} &
        \includegraphics[width=0.20\linewidth]{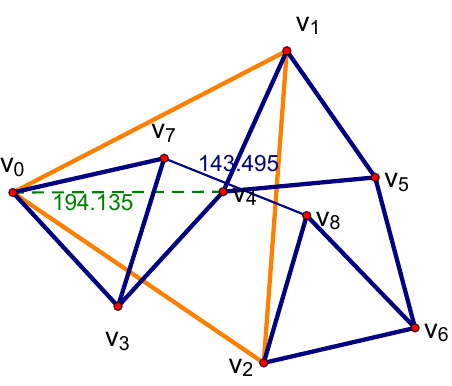} \\
        \includegraphics[width=0.20\linewidth]{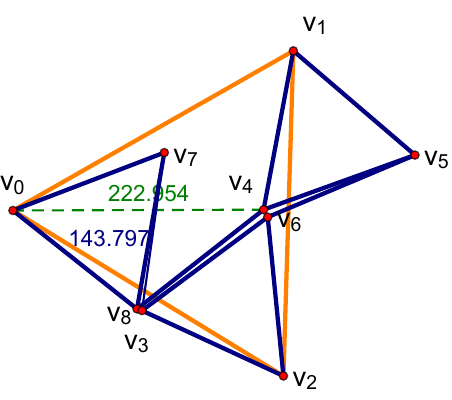} &
        \includegraphics[width=0.20\linewidth]{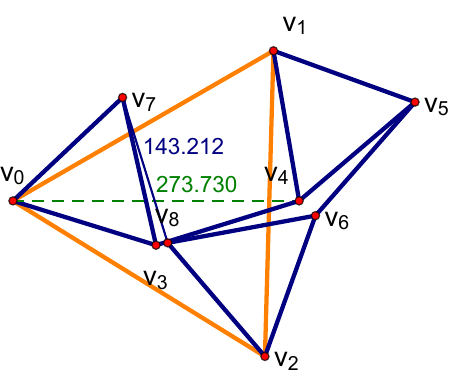} &
        \includegraphics[width=0.20\linewidth]{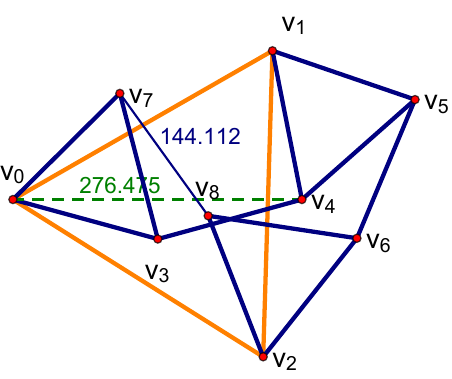}
    \end{tabular} \\ \hline
    \( 3,\,4,\,8 \) & \( \left( v_0,v_4 \right) \) &
    \begin{tabular}{cc}
        \includegraphics[width=0.25\linewidth]{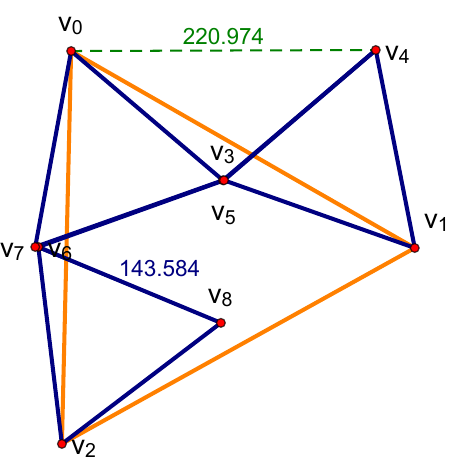} &
        \includegraphics[width=0.22\linewidth]{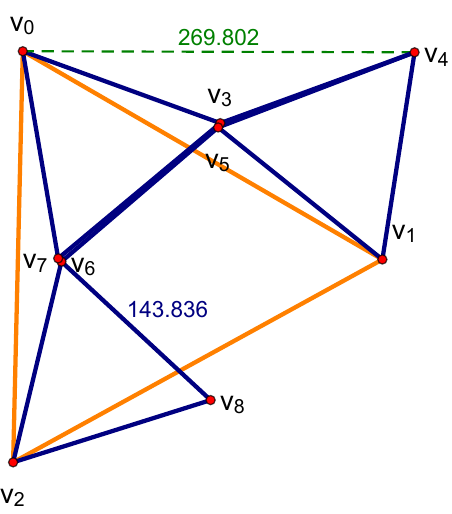}
    \end{tabular} \\ \hline
\end{tabular}
\end{table}

\begin{table}
\label{tab:trihex2_sol}
\caption{\footnotesize 4 solutions found for trihex \ref{fig:trihex2}}
\begin{tabular}{c|c|c}
    \hline
    Flip & Base Edge & Solutions \\ \hline
    \(-\)& \( \left( v_0,v_4 \right) \) &
    \begin{tabular}{cc}
        \includegraphics[width=0.33\linewidth]{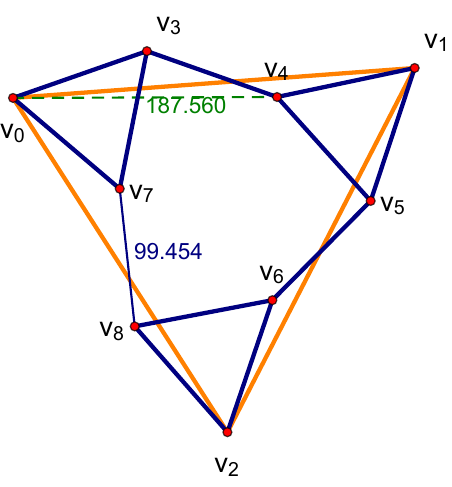} &
        \includegraphics[width=0.33\linewidth]{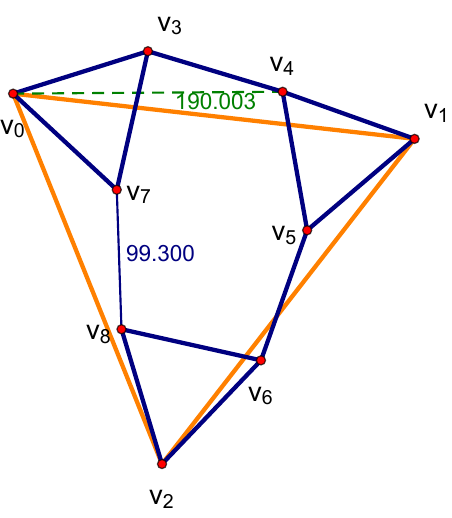}
    \end{tabular} \\ \hline
    \( 3 \)& \( \left( v_0,v_4 \right) \) &
    \begin{tabular}{cc}
        \includegraphics[width=0.33\linewidth]{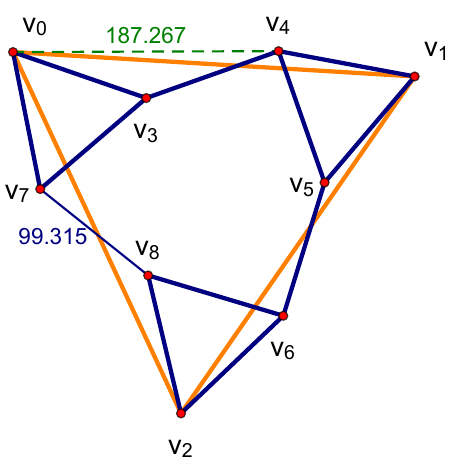} &
        \includegraphics[width=0.33\linewidth]{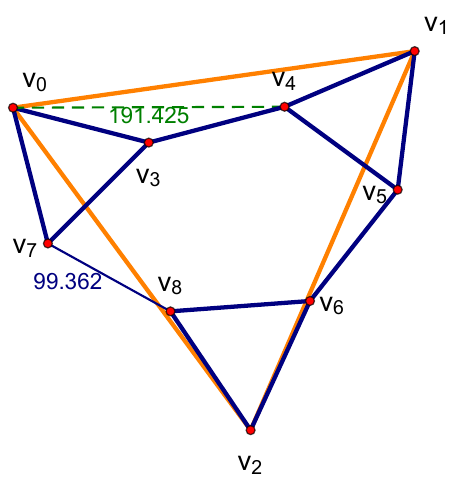}
    \end{tabular} \\ \hline
\end{tabular}
\end{table}

The next method extends and combines two well-known algorithms. The first is the so-called Decomposition-Recombination(DR)-planning algorithm \cite{Baker2015,survey,frontier,feature} that recursively decomposes minimally rigid, or isostatic graphs with edge-length constraints so that only small constraint systems need to be solved simultaneously, and
recombines their solutions \cite{sitharam2006wellformed,sitharam2010optimized,sitharam2010reconciling}  to get all equivalent frameworks. The method of recombination provides an indexing of frameworks around critical configurations \cite{acmtog}.
The second is the use of convex Cayley parameterizations \cite{SiGa:2010}.
The existence of a convex Cayley configuration space is a robust property of graphs underlying frameworks. For frameworks in 2 dimensions, the graphs are so-called partial 2-trees. Informally, complete 2-trees are constructed by pasting triangles together on edges, and partial 2-trees are  subgraphs of complete 2-trees. Such characterizations exist even for frameworks whose bar lengths are in non-Euclidean, polyhedral norms, and is strongly linked to the concept of flattenability \cite{BeCo:2008,SitharamWilloughby2015} of graphs, characterized by forbidden minor subgraphs, for example, partial 2-trees are exactly those graphs that forbid the complete subgraphs on 4 vertices, or $K_4$.
Convex cayley parameterization has been used for analyzing sphere-based assembly configuration spaces in \cite{PrabhuEtAl2020, OzkanEtAl2018, OzkanEtAl2020}, further applied to predicting crucial interactions in virus assembly in \cite{WuEtAl2012,  WuEtAl2020}.

The idea is to drop sufficiently many edges from a glassy framework (\emph{arbitrarily large versions of the Trihex}), shown in red in Figure \ref{fig:drplan-construction}, so that the remaining graph has a convex Cayley configuration space (is a partial 2-tree) parameterized by non-edge lengths shown in green. In this case, the boundary vertices are not pinned, instead boundary edges are judiciously added to make the graph minimally rigid or isostatic, while ensuring the convex Cayley property. This can be achieved for a class of planar graphs encompassing corner-sharing triangular or glassy graphs.
Figures \ref{fig:drplan-construction} and \ref{fig:my_label} have light colored lines  showing the chosen boundary edges.  Importantly, the resulting 2-tree has a simple DR-plan, called a \emph{flex DR-plan} since it requires dropping and adding edges.  Furthermore, the flex DR-plan   permits sequential solving of univariate quadratic equations corresponding to the dropped edges one by one, i.e. \emph{flex 1}, except for a single, final high degree univariate polynomial. This further permits  the indexing of equivalent realizations around critical configurations.
 These key ideas are based on graph theory (to find the flex 1 DR-plan) and numerical solution of a system of quadratic equations by solving a sequence of univariate quadratic equations (with a Cayley parameter as variable) followed by a single univariate equation of large degree (in a final Cayley parameter).   The resulting algorithm to find the solution corresponding to a given index runs in polynomial time in the size of the glassy system and the required accuracy and is fleshed out in upcoming manuscripts \cite{kang1,kang2}.

\begin{figure}
    \begin{subfigure}[t]{\linewidth}
    \centering
    \includegraphics[width=\linewidth]{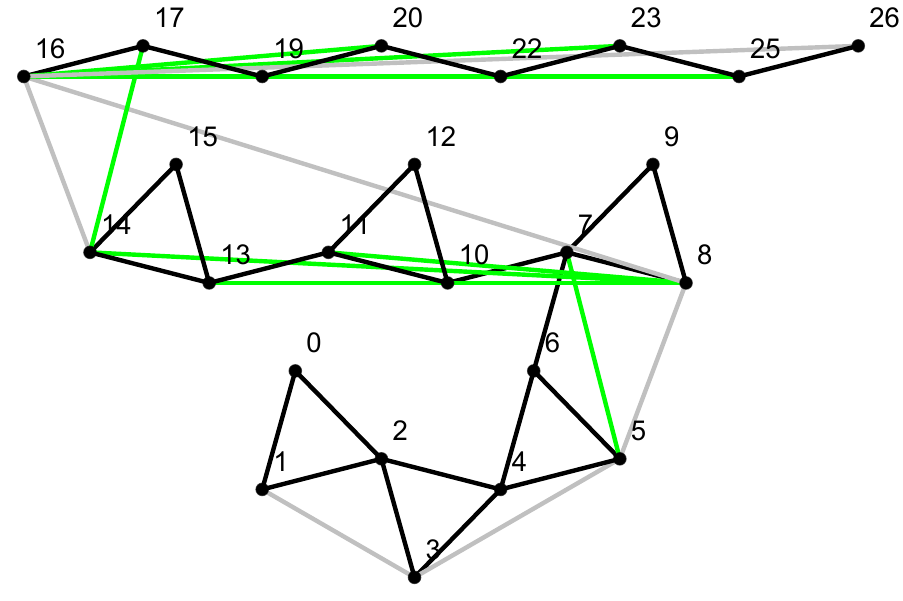}
    \caption{\footnotesize Two-tree with Cayley parameters.}
    \end{subfigure}
    \begin{subfigure}[t]{\linewidth}
    \centering
    \includegraphics[width=\linewidth]{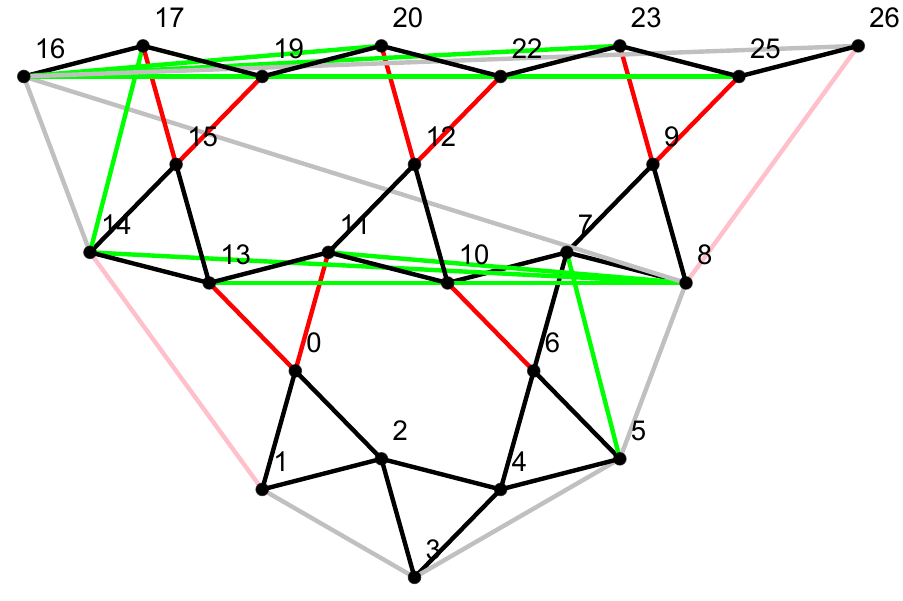}
    \caption{\footnotesize Flex DR-Plan with dropped edges.}
    \end{subfigure} \\
    \caption{\footnotesize We construct the flex 1 DR-plan of a hexagonal lattice manually by applying two actions to the graph: (1) add boundary edges (gray) so that the graph is infinitesimally rigid; (2) drop some of the edges and add some Cayley parameters (green) so that the new graph becomes a two-tree. During recombination, we  add back those dropped edges (red, or pink if boundaries), i.e., we solve for the green Cayley parameter lengths that achieve the given red (and black) edge lengths. An important property of a DR-plan of flex 1 is that we are able to solve for one red edge at a time. Empirically, we usually apply three actions by starting from a small subgraph and expand it by adding one Cayley parameter and drop one original edge. When reach the boundary we add a boundary edge if we cannot construct a two-tree by adding only one Cayley parameter.}
    \label{fig:drplan-construction}
\end{figure}

\begin{figure*}
    \centering
    \includegraphics[width=\linewidth]{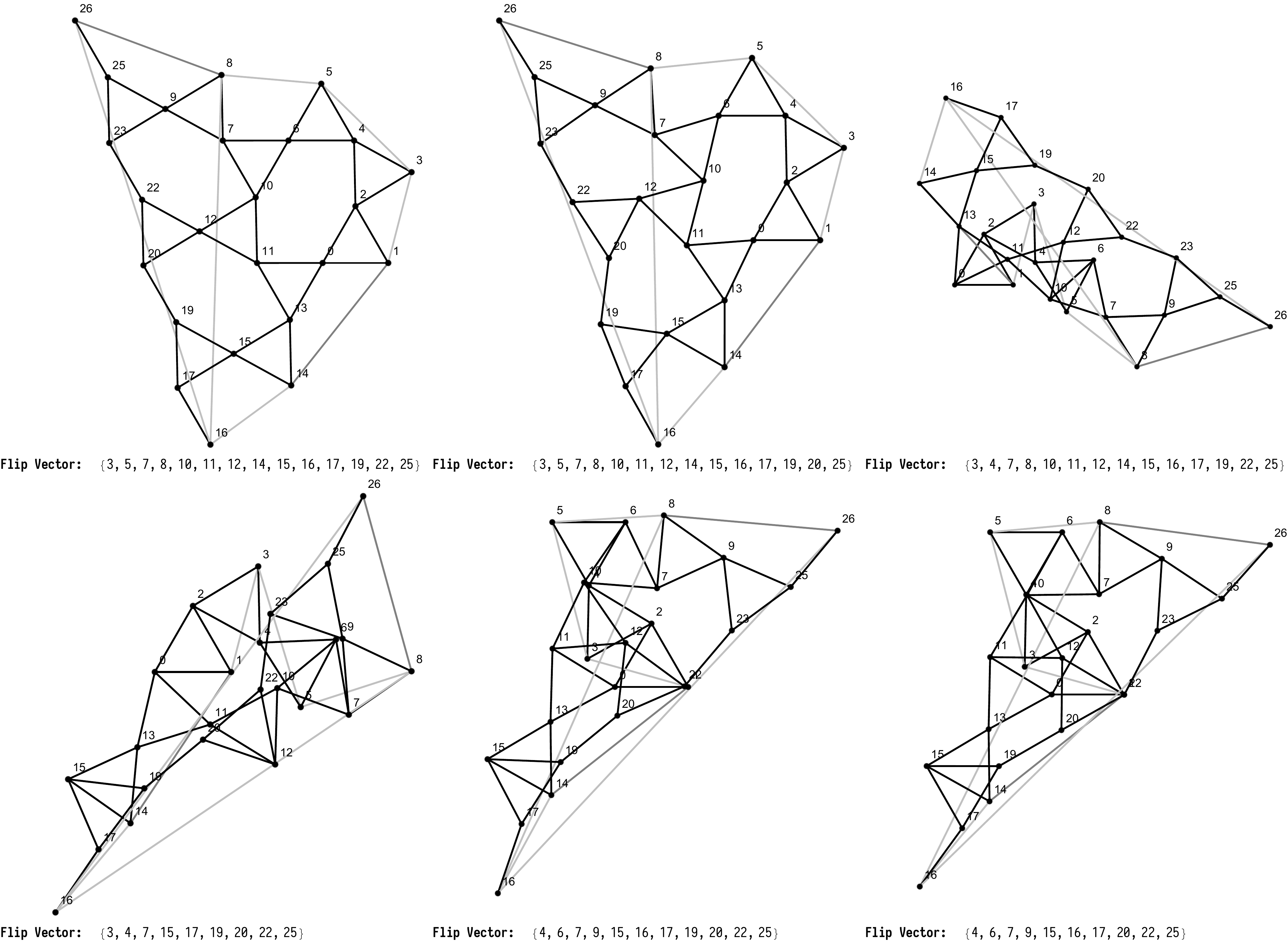}
    \caption{\footnotesize The \(6\) realizations with different flip vectors we found by the DR-Plan solver given Fig. \ref{fig:drplan-construction} as the input. The flip vectors are shown below each realization. The overall error of the dropped edges is $ 0.07 \% \pm 0.21 \% $}
    \label{fig:my_label}
\end{figure*}

\section{Changing ratio of edge-lengths to distance between pinned boundary vertices}
It is trivial that if edge lengths are chosen so short, the triangles
cannot span the distance between the pinned vertices. Therefore at that limit,
set of Eqs.~\ref{eq:edgelens} have no solution.
If we take the vertices with coordinates shown in Fig.~\ref{fig:hexagon}
as an initial structure (where the distance between the pinned vertices is $l_0 \approx 1.0$; recall that the pins are not located on an equilateral triangle and therefore $l_0$ is in fact the average distance between the pins), we are interested in counting the total number of realizations
for a given edge length. Since the Trihex is a small enough example, the complexity advantage of the methods of the previous sections, e.g CayMos are not as crucial. So the realizations are found by directly solving Eqs.~\ref{eq:edgelens} for
uniform edge lengths ranging from $0.34$ to $2$. The equations in \ref{eq:edgelens} are solved in Mathematica using ``NSolve'' function. The pinned boundary
conditions is convenient since no trivial translation or rotation is present~\cite{theran2015anchored}.
Figure~\ref{fig:nrealize} shows the number of \textit{real} solutions for a given $s$.  The red line shows the
total number of distinct \textit{complex} solutions for the Trihex which is fixed and equal to $112$ computed using Magma~\cite{MR1484478}. This number is independent of the chosen edge length and is the upper bound on the number of realizations. By changing the edge length, some but not all of complex solutions become real.

The first real solutions appear to be a single point at
edge length $\approx 0.346$ (see Fig.~\ref{fig:nrealize}). This is an interesting point as it seems there exists only one solution and the
theorem is violated but in fact there are two solutions at this limit although infinitesimally close. This is the signature of a fully stretched network which has the maximum possible volume or the lowest density. From this point of view, the problem is also related to the flexibility window in glasses where naturally-occurring glasses are found near their low-density limit~\cite{sartbaeva2006flexibility,kapko2010flexibility}.
As we increase the edge length $s$, the two infinitesimally close solutions diverge and quickly two new solutions join the previous solutions. The number of
realizations in Figure \ref{fig:nrealize} generally increases up to a maximum number. In fact, two sharp increases happen at $\approx 0.5$ and $\approx 1.0$ due to the fact the pins are roughly $l_0 \approx 1.0$ unit
apart. Therefore when $l_0$ is roughly an integer multiple of
the edge length of the triangles, the triangles can tightly fit together and new solutions appear. Figure \ref{fig:sols} shows some realizations (out of 76 possible realizations) for the edge length $1.0$.

After reaching a maximum of $104$ realizations, the number of solutions rapidly drops. Our numerical experiments show that a subset of solutions survive even at very large edge lengths (high density) and the number of realizations reaches a plateau of $44$ solutions. Fig.~\ref{fig:limit} lists these $44$ states when the edge length is set to $100$ but Fig.~\ref{fig:nrealize} shows the number of realizations up to $s=2$. Many solutions in this regime are related by
an approximate mirror symmetry, as we expect the three blue pins to be co-incident
and the equilateral triangles are roughly arranged around a central point.

\begin{figure}[t]
\centering
\includegraphics[scale=0.4]{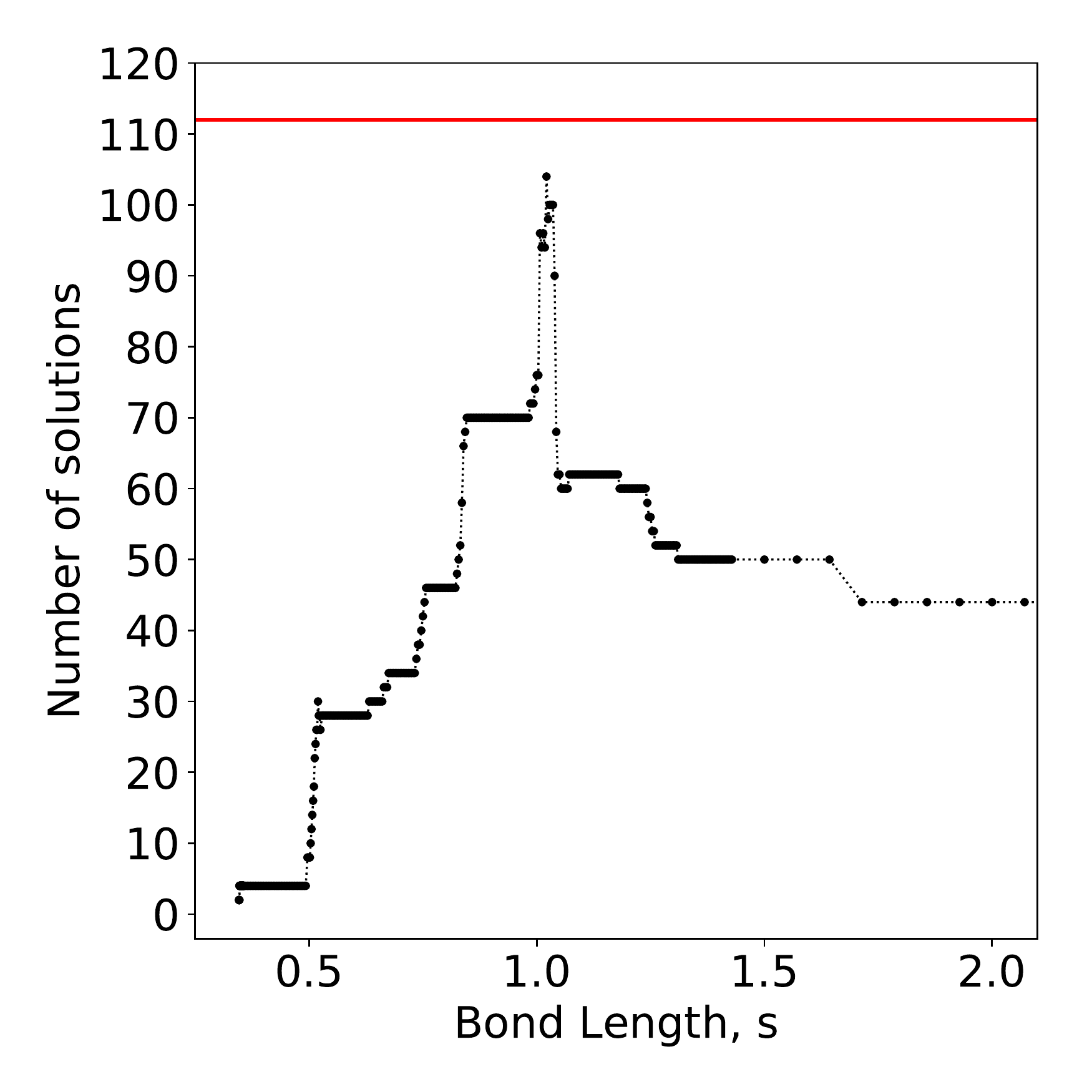}
\caption{\footnotesize Number of realizations of Trihex shown in Fig. \ref{fig:hexagon} by varying edge length in the units that the distance between pins $l_0=1$. The horizontal red line shows the total number of {\it{complex}} solutions which is equal to $112$. The number of solutions increases rapidly when the edge lengths are half and equal to the pins distance.}
\label{fig:nrealize}
\end{figure}

\begin{figure}[t]
\centering
\includegraphics[scale=0.3]{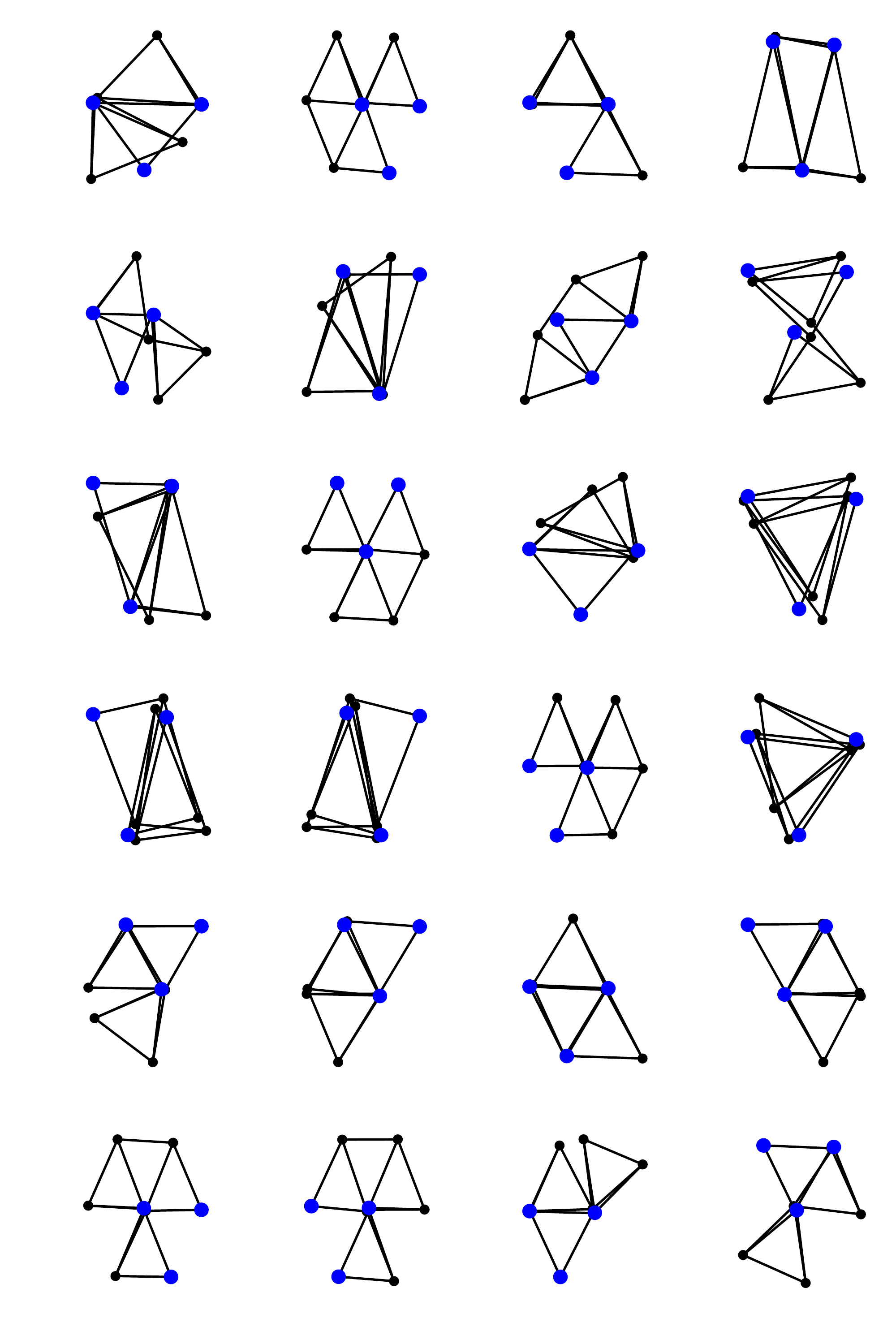}
\caption{\footnotesize Some of the realizations with the edge length equal to $s=1.0$.}
\label{fig:sols}
\end{figure}

\begin{figure}[t]
\centering
\includegraphics[width=7.5cm]{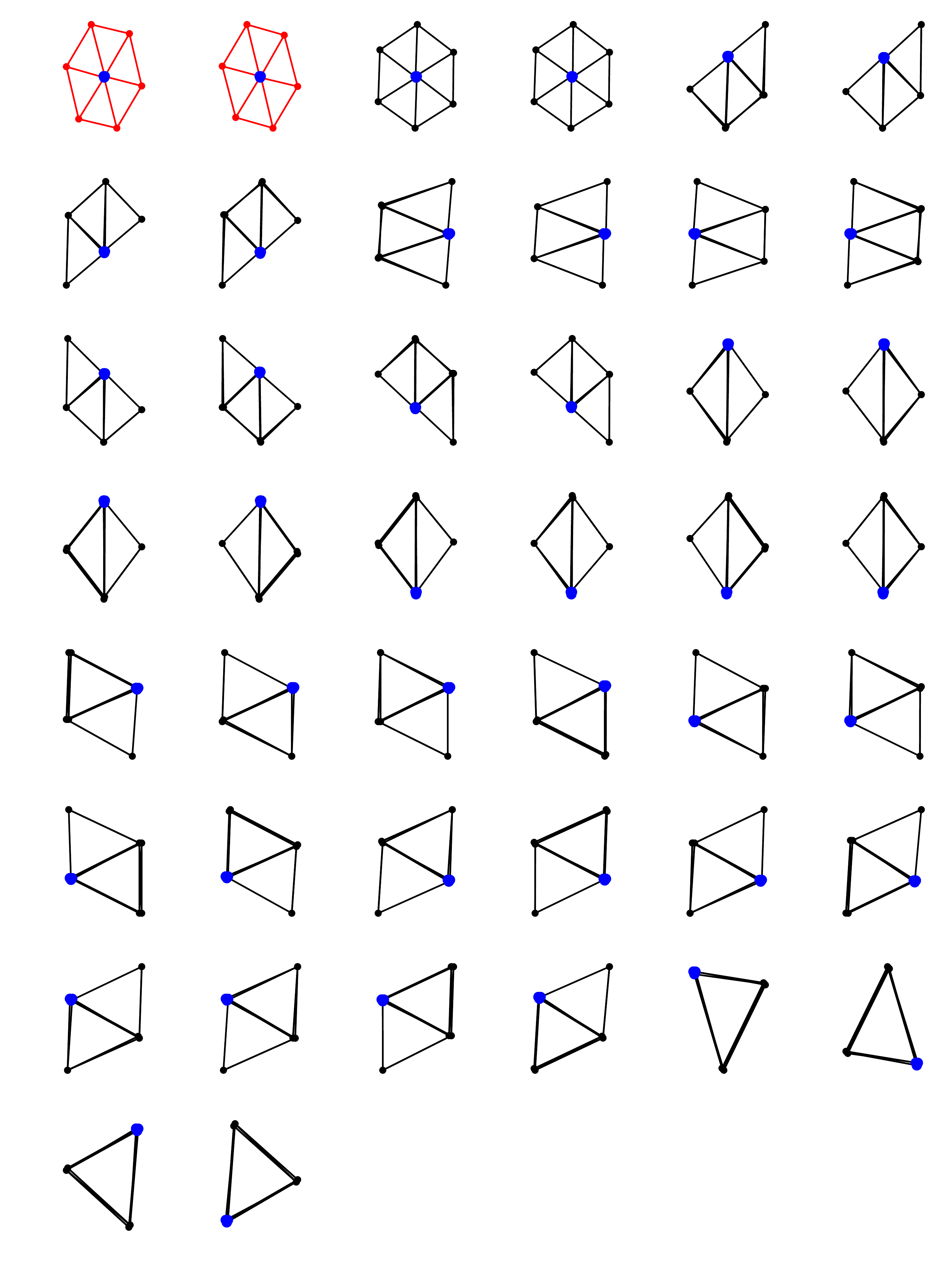}
\caption{\footnotesize The 44 solutions in the large edge length limit where the edge length $s$ is set to 100. Similar solutions are related by mirror symmetry but not rotation. The first two realizations are colored to emphasize that although the graphs look alike, different vertices have occupied the same coordinates and therefore they count as two distinct realizations.}
\label{fig:limit}
\end{figure}

Although it is pedagogic to inspect the individual solutions visually, we need to
distinguish various realizations with the same edge lengths. This also helps to track how
solutions at a given edge length from the previous solutions with smaller edge lengths. We choose
the mean distance of vertices from the centroid of the pinned vertices as the metric.

If we plot this metric versus the edge length ratio for each realization, the result is the \emph{trajectories} in Figure \ref{fig:lbar}, showing how some solutions persist for a long range but
others disappear. These trajectories represent solutions to a different set of equations than the $2N -1$ equations of the form Equation \ref{eq:edgelens}   whose solutions form the closed curves of the 1-dof  mechanisms given by the single-cut or Caymos algorithms.
This system has $2N$ equations of the form Equation \ref{eq:edgelens} for the bonds,  but has an extra variable $s$ - representing the bond length  (boundary-edge length is fixed due to pinned vertices).
Red and green lines respectively show the linear and quadratic fit to the
persistent paths which shows an intermediate growth rate. Previously, we discussed
that there is a sharp increase in the number of realizations at $s$ around $0.5$ and $1.0$.
Fig.~\ref{fig:lbar} shows that along those values, there is a tremendous amount of
activity and a large set of solutions are only present in a small region of
edge lengths.


\begin{figure}[!ht]
\centering
\includegraphics[scale=0.4]{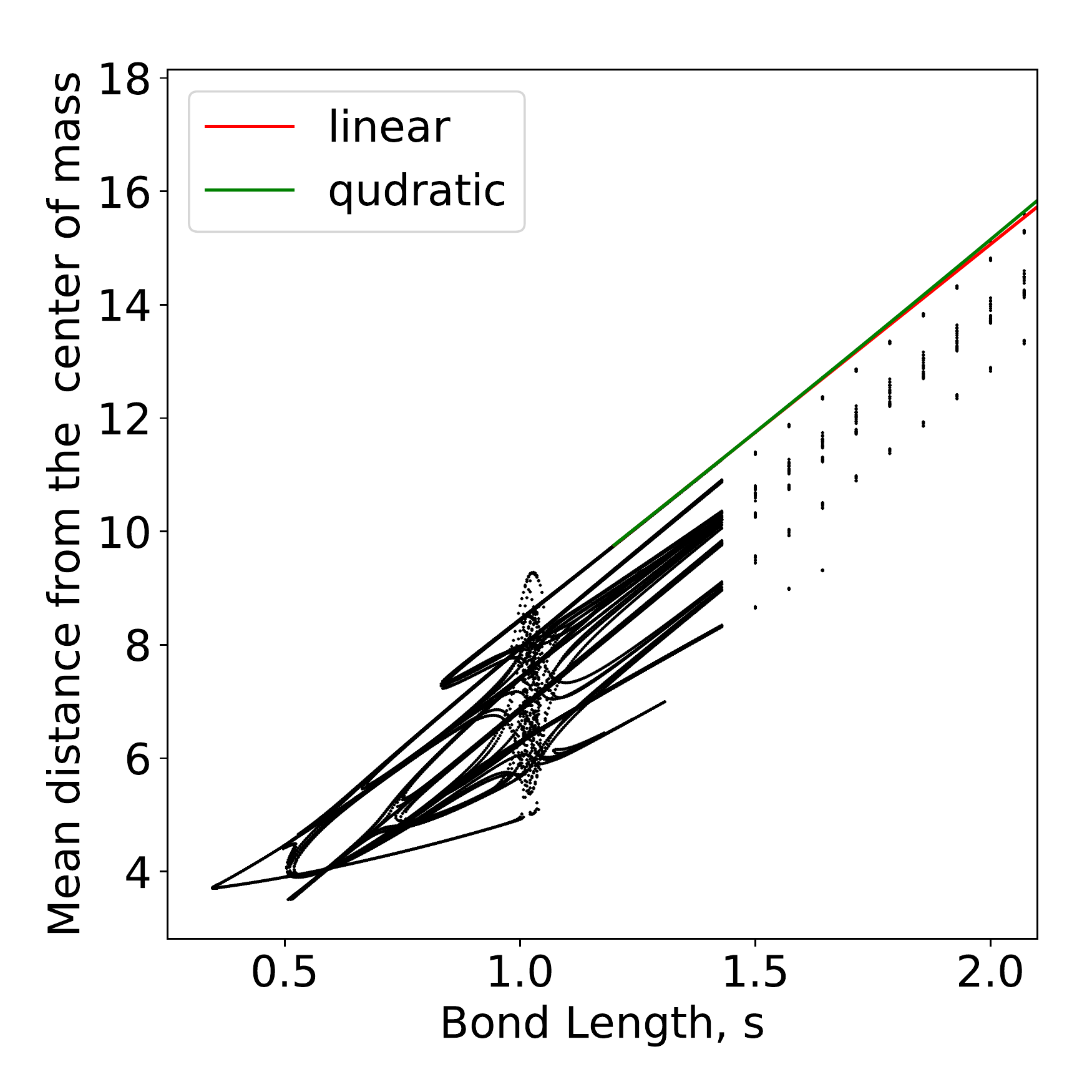}
\caption{\footnotesize The mean distance of vertices from the centroid of the pinned
vertices vs. the edge lengths. The mean distance scales quadratically (green) not linearly (red). The two curves are fitted to the
topmost points with edge lengths $s$ between 1.1 and 1.4 but are extrapolated to the
outside of this window.}
\label{fig:lbar}
\end{figure}

\begin{figure}[!ht]
\centering
\includegraphics[scale=0.3]{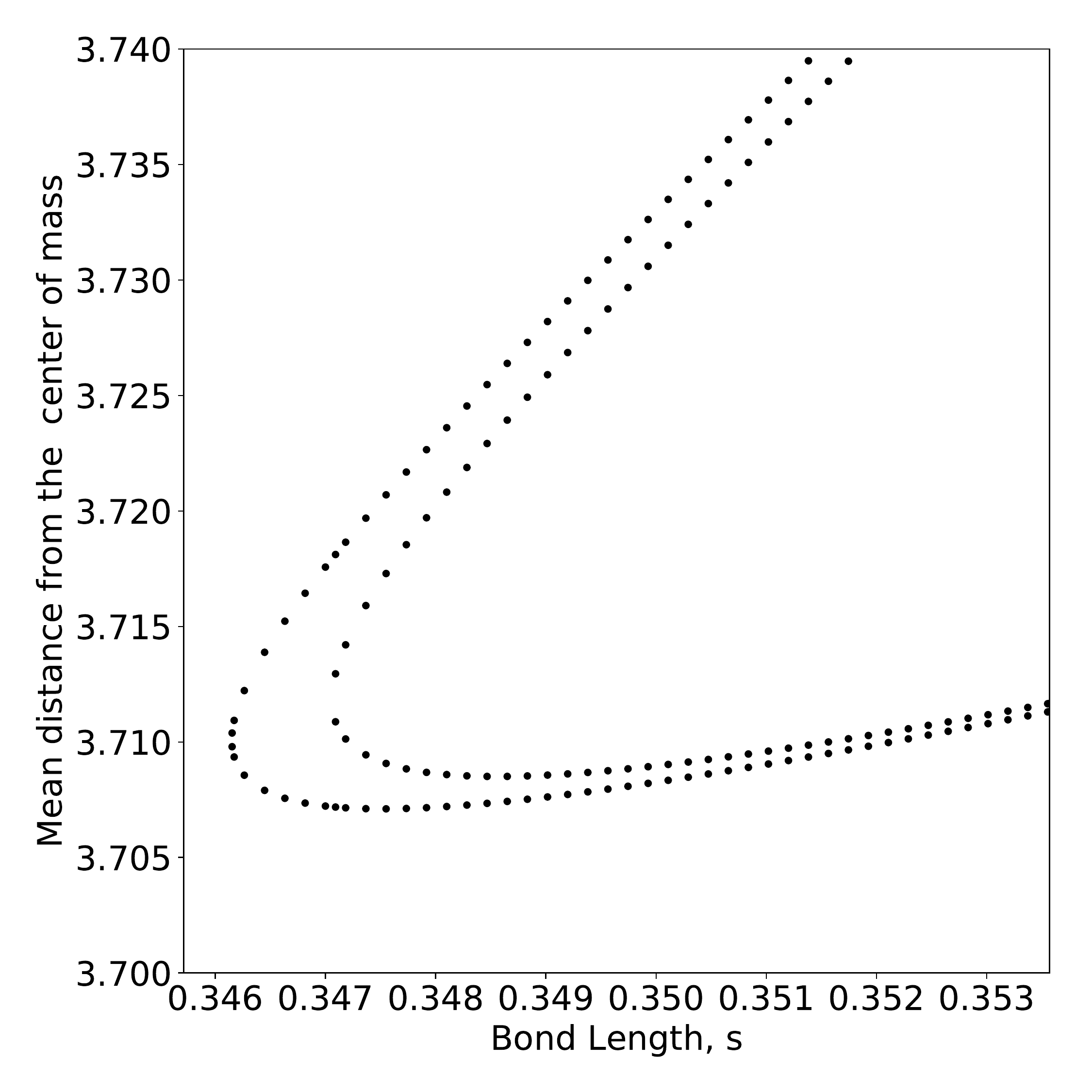}
\includegraphics[scale=0.3]{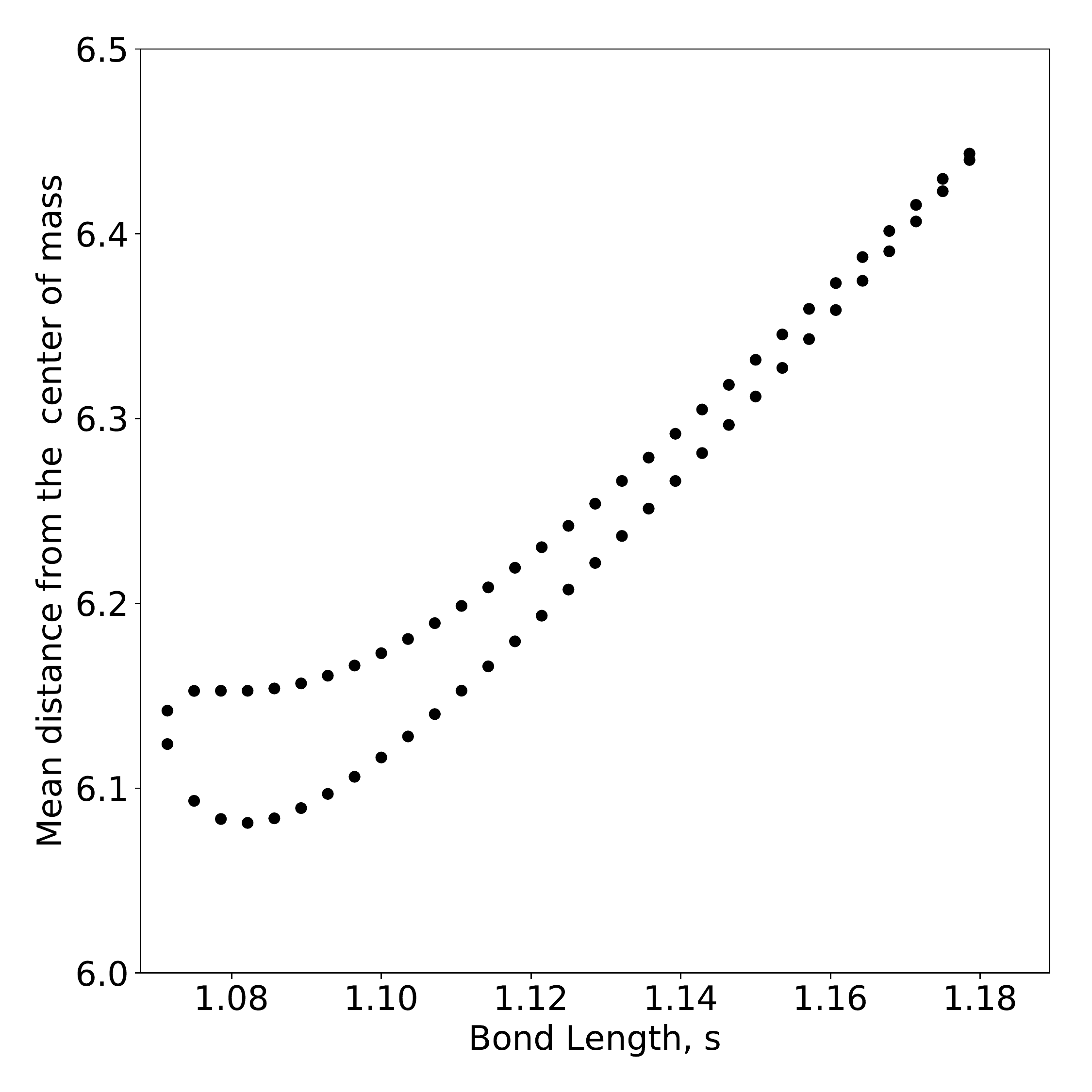}
\caption{\footnotesize
Open  trajectories (top) and simple closed trajectories (bottom) are two ways in which realizations appear and disappear. The upper panel depicts the initial solutions at the low-density limit. Note that first two solutions emerge and then they diverge while at a secondary point, a new trajectory of solutions appears. In the lower, a pair of solution gradually converge and finally form a close loop.}
\label{fig:initial}
\end{figure}

\begin{figure}[ht]
\centering
\includegraphics[width=8cm]{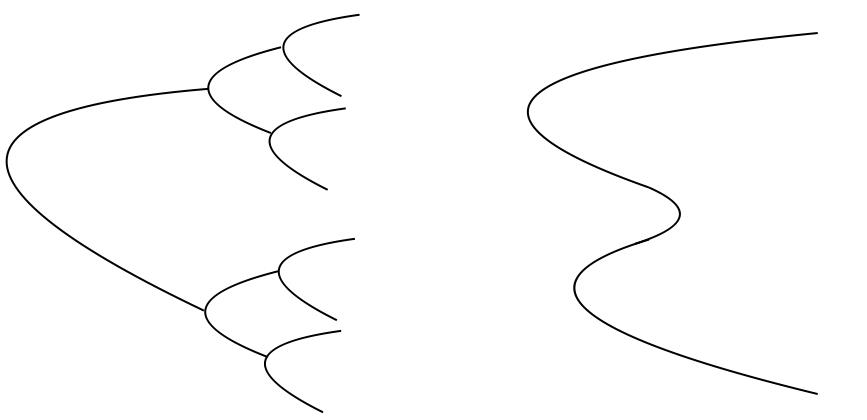}
\caption{\footnotesize The left panel is a bifurcation which we have never observed. The right panel is a ``retrograde'' which is an alternative way of losing solutions, which we do observe for the trihex.}
\label{fig:retrograde}
\end{figure}

The complexity of paths in Fig.~\ref{fig:lbar} makes it certainly constructive to
look at specific regions of edge length in more details.
Our observations show that new solutions always come in as a pair.
Based on the results, we have observed three mechanisms for
appearance/disappearance of solutions: \textit{simple closed trajectories},
\textit{open trajectories}, and \textit{retrogrades}. Example of
simple closed and open trajectories are given in Fig.~\ref{fig:initial}. Open trajectories are the persistent  trajectories that continue to exist even at very large edge lengths. A retrograde is
a  trajectory that bends backward which is a disappearance mechanism; an example is given
in the right panel of Fig.~\ref{fig:retrograde}. This leads to more complex circuits replacing the simple loop in upper part of Fig.~\ref{fig:circuits}.  For the retrograde there can be 4 intercepts and it is clear that any closed loop will have an even number of intercepts, consistent with Theorem 1.
However, we have found no evidence of
\textit{bifurcations}~\cite{arnol2003catastrophe,feigenbaum1978quantitative}(Fig.~\ref{fig:retrograde}, left panel).

\begin{figure}[t]
\centering
\includegraphics[scale=0.3]{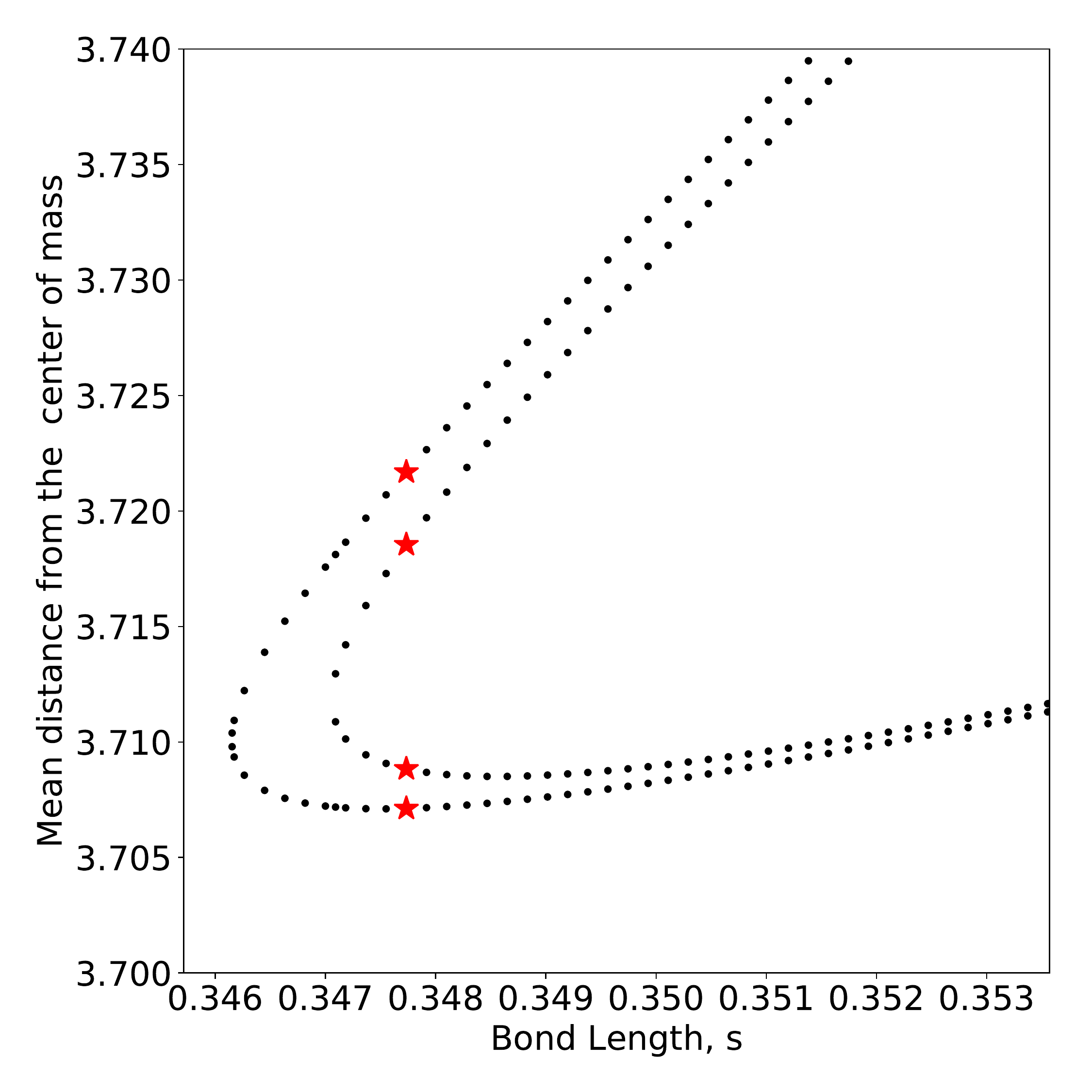}
\includegraphics[scale=0.25]{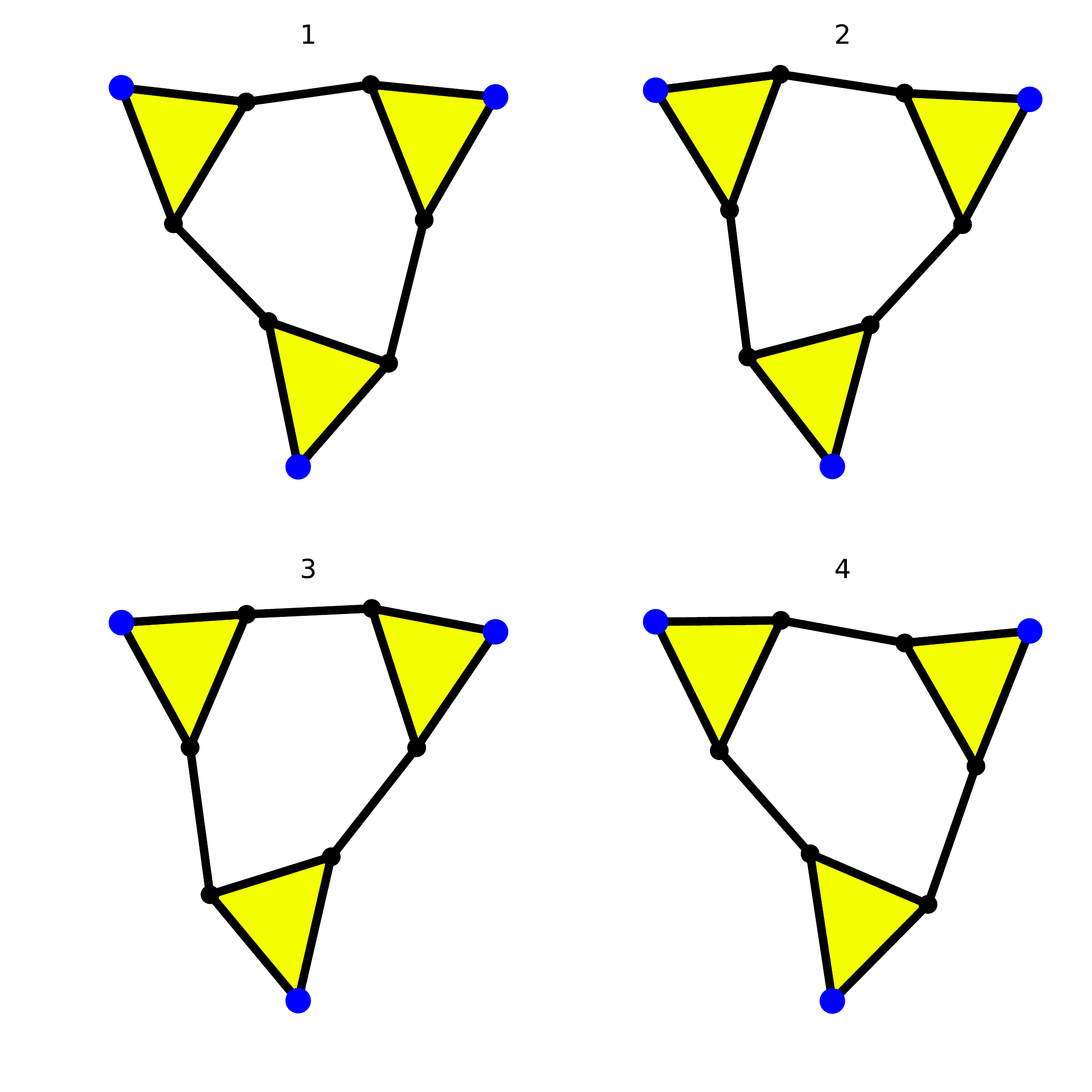}
\caption{\footnotesize The four solutions with the edge length equal to $\approx0.3477$ marked with
red asterisks. The solutions are
numbered from the smallest mean distance from the centroid ($y-$axis) to the largest.
The solutions on the same trajectory show a small displacement but a more significant
motion is involved among the solutions from the different trajectories.}
\label{fig:4sols}
\end{figure}

\section{Energy Barriers}
The discussion in the previous section showed that different realizations of an
isostatic network can be found by solving edge length equations.
Realizations of a framework are thought to be related to tunneling states in glasses. Fig.~\ref{fig:4sols}
depicts four realizations of Trihex at $s \approx 0.348$. Their corresponding
points are marked by red stars in the left plot.
The amount of energy in transition from one state to another is a way  of classifying sets of states.
Note that the landscape has no minimum except the listed four
solutions.
We perform a linear interpolation between two given states. Let $\mathbf{s}_1$
and $\mathbf{s}_2$ be two solutions in the configuration space. We write:
\begin{equation}
  \label{eq:linearinter}
 \mathbf{s}(\lambda) = \mathbf{s}_1 + (\lambda+\frac{1}{2}). (\mathbf{s}_2-\mathbf{s}_1)
\end{equation}
Assuming bonds are harmonic springs with spring constant equal to unity, the energy can be
found as a function of $\lambda$, where $\lambda = -\frac{1}{2}, \frac{1}{2}$
correspond to the two states.
We can think of two paths on this plot
as two   trajectories of frameworks. The equivalent frameworks $1$ and $4$ belong to the  first
 trajectories while $2$ and $3$ lie on the second trajectory which exists only  when $s \ge 0.347$.
The pair of  frameworks that belong to the same trajectory 
indeed have very similar
configurations. In fact, the largest difference between the pairs is the reflection of
the top connecting edge along horizontal axis. For the pairs that do not belong to
the same  trajectory,
the motion involves the significant
rotation of the bottom triangle.
If we would assume that the edges are \textit{harmonic springs} and not fixed-lengths bars,
the energy path connecting the pairs of realizations on different branches has a much higher
energy barrier compared to that of the pairs on the same branch. Note that the whole energy landscape of Trihex at this edge length has only $4$ minima. The
nice feature of the landscape of Trihex is that the global minimum energy is
exactly zero.
\begin{figure}[ht]
\centering
\includegraphics[width=8cm]{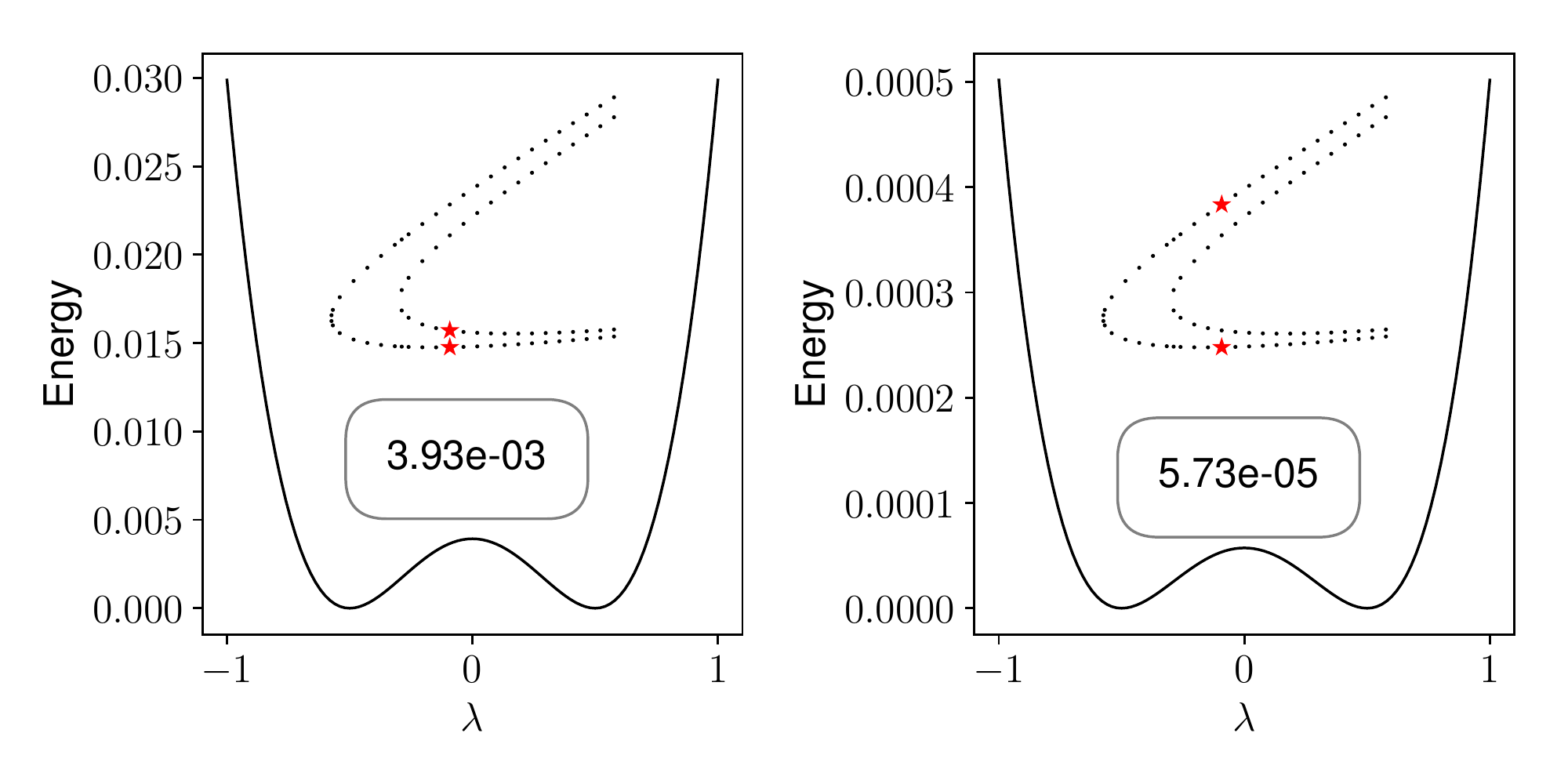}
\caption{\footnotesize The energy landscape of transition states between various realizations
of Trihex, found by linear interpolation (Eq.~\ref{eq:linearinter})
between the realizations shown in Figure~\ref{fig:4sols}.
The inset box shows the height of the energy barrier in the units that spring
constant is set to unity. The inset figure shows the two solutions indicated by
red asterisks. The energy barrier between the realizations on the same branch
is significantly smaller than that of the realizations on different branches.}
\label{fig:barrier}
\end{figure}

This energy perspective makes an important bridge between Trihex examples and
glasses. If this picture from studying Trihex remains intact in glasses, we
expect to see that solutions play different roles depending on which branch
they belong. Fig.~\ref{fig:4sols} is particularly important since
the experimental density of glasses is close to the low density edge. So we
expect the discussion in this section would somewhat generalize to the 2D
glasses.
But as discussed, solving edge length equations is computationally expensive for a
large system. On the other hand, two-level systems in glasses are rare. To have the
slightest hope of finding a tunneling state, we need to study systems that are
considerably larger than Trihex.
This makes it inevitable to design an alternative approach to find realizations of
a framework starting from already available information.

\section{2D Glasses and Jammed Networks}
Trihex serves as an illuminating toy model to present the main ideas about the existence of multiple realizations for an isostatic network and techniques to find such realizations. However, these methods are categorically applicable to the networks inspired by or modeled directly from materials such as network glasses or jammed granular packing. Such materials are either at or very close to the isostatic states and their behavior is significantly influenced by their rigidity~\cite{sartbaeva2006flexibility,kapko2010flexibility,ellenbroek2015rigidity}. But our main concern here is to establish a link between the multiplicity of material realizations and their physical properties, specifically the existence of the tunneling modes in such materials.

In the case of glasses, we are focused on those glasses that can be modeled as a network of corner-sharing tetrahedra in 3D or triangles in 2D. An example of the former is SiO$_{2}$ or GeO$_{2}$ and of the latter is a silica bilayer which although is 3D dimensional but it can be seen as two mirroring layers of 1-atom thick of Oxygens connected through bridging atoms to complete the chemical bonds~\cite{sadjadi17refining}.

In 2D glasses, every atom/vertex is four-fold coordinated (four shared constraints) but since each vertex has two degrees of freedom (translational degrees of freedom), 2D glasses are \textit{locally} isostatic~\cite{theran2015anchored}. However, boundary conditions determine the \textit{global} rigidity of the framework. For networks extracted from the experimental data, atoms on the surface are not fully connected and anchored/pinned boundary conditions are necessary and sufficient to render the system isostatic~\cite{theran2015anchored}. For material networks made using computer simulations, the boundary conditions are periodic which means such 2D systems contain two redundant bonds which must be removed to render the network isostatic.

These computer-generated networks are prepared using Wooten-Winer-Weaire (WWW) algorithm
with the periodic boundary conditions while ensuring that the ring distribution
and the area of polygons are in agreement with the experimental data~\cite{wooten1985computer, kumar2014ring}. As a consequence, the edge lengths are no longer exactly equal and 2D glasses satisfy
a stronger definition of being \textit{generic}. Similar to Trihex, the structure is in
mechanical equilibrium, all edges are assumed to be harmonic springs initially at their rest lengths
and the dynamical matrix is positive semi-definite.

In the case of granular networks, grains are often modeled as an athermal packing of circles/spheres interacting via a Hookian or Hertzian potential. To model such packings we use the standard protocols that are common in jamming community. We start with a random distribution of bidisperse $(0.5:1,0.5:1.4)$ circles (to prevent crystallization) and we rescale all the radii uniformly to set any desired packing density $\phi$. The energy of the system, given by Eq. (\ref{eq:jammedenregy}), (where $\rho_{ij}$ is the distance between nodes $i$ and $j$ and $\sigma_{ij}$ is the sum of their radii) is then minimized by a standard FIRE algorithm~\cite{bitzek2006structural} using the pyCudaPacking package, developed by
Corwin \textit{et al.}~\cite{morse2014geometric, charbonneau2012universal} :

\begin{equation}
    \label{eq:jammedenregy}
    E = \sum_{ij} (1-\frac{\rho_{ij}}{\sigma_{ij}})^2 \ \Theta(1-\frac{\rho_{ij}}{\sigma_{ij}}),
\end{equation}
where $\Theta$ is the Heaviside step function to ensure that only overlapping circles are included in energy.

If the packing density is high enough ($\phi > 0.84$), the system will have several states of self-stress or redundant contacts. By decreasing the packing density quasi-statically, the system reaches a critically jammed state with zero pressure and one state of self-stress. This system can then be mapped to a network by replacing the center of mass of each circle with a node and replacing any none-zero overlap between neighboring particles with a bond between their corresponding nodes. Such a network has one bond in excess of isostaticity; By removing any one bond randomly one can make an isostatic network.

Once a system is at the isostatic point, in principle, the same techniques employed to find realizations of a Trihex are also applicable to material networks. However, such networks differ from Trihex in two significant ways. Firstly, material networks contain many possible atomic arrangements which leads to various couplings among atoms in the set of edge length equations (Eq.~\ref{eq:edgelens}). While Trihex is essentially a ring of triangles forming a hexagon, in experimental samples ring size varies from 4-8~\cite{kumar2014ring} which are distributed non-randomly on a plane~\cite{Sadjadi2016}.

Secondly, 2D glasses are generally much larger than Trihex and it might not be computationally feasible to apply the techniques directly. In fact, solving the set of edge length equations (Eq.~\ref{eq:edgelens}) exactly is practically impossible for systems as large as 2D glasses with $N\geq\mathcal{O}(10^2)$ which in turn means for larger systems the complete set of solutions and their evolving on branches are inaccessible.

The alternative method, the single-cut algorithm, is guaranteed to provide new realization(s) but is not an exhaustive method and some of the existing solutions will be unreachable. In this method, we need to compute the null space (the eigenvectors corresponding to zero eigenvalue) of a matrix of size $(dN)^2$ where $d$ is the spatial dimensions and $N$ is the number of vertices. Even if finding the null space can be done efficiently by avoiding diagonalizing the entire matrix, there is some error associated with moving $N$ atoms along non-trivial zero-mode. It is possible that the path would not be closed (the system would not return to its original conformation) due to accumulation of errors. Therefore, either the step size should be chosen sufficiently small to ensure the path is smooth or frequent energy-minimizations are necessary along the path; both of which are expensive.

\begin{figure}[t]
\centering
\includegraphics[scale=0.25]{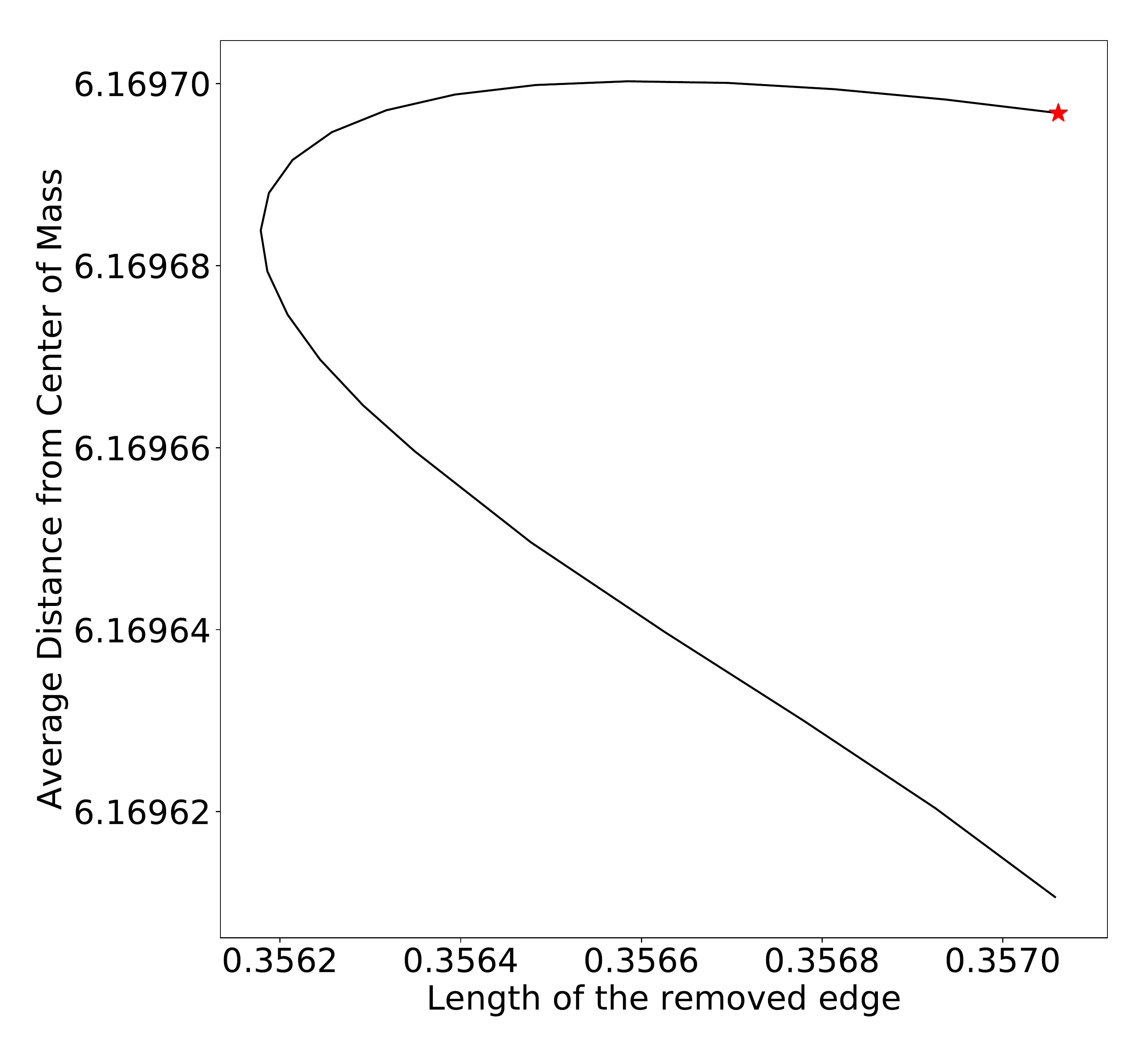}
\caption{\footnotesize The curve found by applying the modified single-cut algorithm to a 2D glass, in which the path traversal is stopped upon finding a solution. Fairly large steps are taken along the eigenvector with zero eigenvalue as is evident from the curve roughness.
The vertical axis represents the total distance of all vertices from the center of mass while the horizontal axis shows the distance between two ends of the removed edge. The red asterisk at the bottom denotes the original network and the top asterisk shows the alternative solution found by the path traversal.
The arrow is drawn to emphasize the fact that the real solution (indicated by the blue asterisk) with no error in
the edge lengths does not exactly lie on the drawn curve in configuration space and further energy minimization is required to find the correct coordinates.}
\label{fig:circuit300}
\end{figure}

Hence, we slightly modify the single-cut algorithm to find the alternative realizations of large isostatic networks. Since the existence
of the second realization is guaranteed, we take fairly large steps along the non-trivial zero-eigenvalue eigenvector alternative. In addition, it is not necessary to complete the curve in configuration space exactly to calculate the thermal properties. Therefore, once the curve intersects the vertical line denoting
the original length of the cut edge, we have found a second realization and the
path traversal can be stopped early. However, because of the larger steps along the curve, the position of vertices and subsequently the edge lengths have large errors which means
the conformation has not returned exactly to its original set of edge lengths. But since the conformation is close enough to the energy basin, we only refine the last conformation by an extensive energy minimization
to ensure that edge lengths are equal to their original values (Fig.~\ref{fig:circuit300}).

If the curve is complex enough, it contains more than two realizations.
Based on the modified scheme, if the curve traversal is stopped after
finding the first solution, we might miss a whole set of solutions. To address this concern, we note that glasses are found at the
extreme of density (edge of the flexibility window)~\cite{sartbaeva2006flexibility}. At this limit, only two solutions are expected, similar to Trihex example where at the maximal density point only two solutions existed. We tested the validity of this argument by applying the original single-cut algorithm on two 2D networks
($N=48 \textrm{ and } 300$) by removing all edges iteratively. It was observed that all closed curves
give two and only two distinct solutions independent of which bond is removed.

After applying the modified single-cut algorithm, two realizations of a 2D glass are available; they have the same exact topology and bond lengths, but the vertices
are displaced between the two states. The amount of this displacement determine whether the conformations are in fact the tunneling states. Anomalous specific heat is observed at temperatures about $1$ {K}, where the available energy is not sufficient for the displacement of a large group of atoms over a long distance. Therefore, it is expected that atomic displacements in a tunneling state are relatively localized. To characterize to what extent the displacements between two conformations are localized, we use the Participation Ratio (PR). If an atom $i$ is displaced by the vector $\mathbf{u}_i$ between two conformations, PR is defined as:
\begin{equation}
  \text{PR} = \frac{\left(\sum_{i=1}^{N}
  {|\mathbf{u}_i|}^2\right)^2}{N \sum_{i=1}^{N} {|\mathbf{u}_i|}^4}.
\end{equation}
For a perfectly delocalized mode ${|\mathbf{u}_i|} \sim 1/\sqrt{N}$ and $\text{PR} \sim 1$.
For a completely localized mode $|\mathbf{u}_i| \sim \delta_{ij}$, $\text{PR} \sim N^{-1}$.
Hence, a small value of PR is the signature of a localized mode. In the case of tunneling modes, it is expected that by increasing the system size $N$, the fraction of atoms participating in the mode decreases.

In addition to the locality of the displacement, it is essential to measure the significance of the atomic displacements in the limit of large systems. For a conformation to be considered a tunneling state, the atomic displacements should be significantly larger than zero-point motion. Assuming a harmonic oscillator, the zero-point amplitude $x_{0}$ is of order of:
\begin{equation}
  x_{0} \sim \sqrt{\frac{\hbar}{m \omega}} = \sqrt{\frac{10^{-34}}{10^{-26} \times 10^{14}}} = 10^{-11} \textrm{m} = 0.1 \,\textrm{\AA},
\end{equation}
for an oxygen atom. For an O$-$O bond length of $2.6${\AA}, $x_{0} \approx 10^{-2}$ in the unit of the bond length. If the typical
motion of the atoms measured by their mean displacement $N^{-1}\sum |u_i|$ is less than
$x_{0}$, such motions are not relevant to the tunneling states but nevertheless, they are mathematically correct and give rise to other realizations.

To quantify the atomic displacements and propitiation ratio in large systems, we prepare four networks of corner-sharing triangles under
periodic boundary conditions with varying size $N=48, 300, 1254, 5016$ and randomly remove two edges to satisfy the isostaticity condition. By applying the modified single-edge
cut algorithm, the corresponding second realizations are found. This allows us to study the behavior of the participation ratio and the mean displacement
of vertices as a function of the number of particles.

Table~\ref{tab:motionvals} summarizes the results for the mean displacement of atoms
and their participation ratio. The total displacement $\sum |u_i|$ increases slightly by system size but the average displacement of a typical particle $\sum |u_i|/N$ generally decreases. But regardless
of $N$, the mean displacement is smaller than of the zero-point motion $x_0$ and hence this motions cannot be representative of a tunneling state.
In addition, all networks exhibit modes that about $\sim 45$\% of all vertices are displaced in the system. Such an extended mode cannot be a tunneling state since in the limit of Avogadro number of atoms,
a massive number of atoms should be involved in such states which is not
energetically favorable. Unfortunately, it seems that the single-cut algorithm is
not able to find realizations that are sufficiently distant from the initial
realizations (evidenced by vanishingly small $|u_i|$ values) and sufficiently
localized (evidenced by the constant PR$/N$ value).

\begin{table}[tb]
\centering
\caption{The magnitude of displacements in the unit of the edge length found in simulations for different system sizes, $N$.}
\label{tab:motionvals}
\begin{center}
  \begin{tabular}{ c  c  c  c  c  c }
    \hline
    $N$ & $\sum |u_i|$ & $\sum |u_i|/N$ & $\sqrt{\sum |u_i|^2/N}$ & PR \\ \hline
    48 & $0.09$ & $1.89 \times 10^{-3}$ & $2.12 \times 10^{-3}$ & 0.44 \\
    300 & $1.47$ & $4.89 \times 10^{-3}$ & $5.80 \times 10^{-3}$ & 0.41 \\
    1254 & $2.03$ &  $1.62 \times 10^{-3}$ & $1.84 \times 10^{-3}$ & 0.47\\
    5016 & $2.94$ & $0.59 \times 10^{-3}$ & $0.67\times 10^{-3}$ & 0.46\\\hline
  \end{tabular}
\end{center}
\end{table}

Although from the localization and displacement considerations, it is evident that the found realizations cannot account for the tunneling states, nevertheless it would be insightful to study their thermal properties. For each system size, we can form a double-well potential where each realizaion sits at one of the energy minima. For every $N$, the energy pathway is found by the linear interpolation
(Eq.~\ref{eq:linearinter}) betweem the two realizations with zero energy. The height of the energy barrier $V_b$ can be estimated from the interpolated curve while the well separation $d$ is calculated
as the root-mean-square deviation (RMSD) of atomic positions.

Fig.~\ref{fig:dws} shows an example of a double-well potential
derived from the system with $N=300$ atoms. The black points are found using Eq.~\ref{eq:linearinter} while the red line is a $4$th-order polynomial fit to these points~\cite{sadjadidiss}.
The two blue lines are the harmonic approximations around two equilibrium realizations. The probability of the tunneling scales as $e^{-\lambda}$ where $\lambda$ is the tunneling parameter defined by the following equation
(derived from the ratio of kinetic and potential energies):
\begin{equation}
 \lambda = d \sqrt{\frac{2 m V_b}{\hbar^2}},
\end{equation}
where $m$ is the mass of an oxygen atom (See Appendix B in~\cite{sadjadidiss} for the derivation and a detailed discussion on the significance of the tunneling parameter).

Table~\ref{tab:dwpvals} summarizes the characteristics of the double-well
potential for four systems in SI units. The barrier height ($V_b$) of all
systems is a very small value which means the double-well is essentially flat in
the middle. The well
separation $d$ is also very small and at most about 5\% of O$-$O bond length and $\lambda$
shows a somewhat monotonic decrease with the system size (to find the exact dependence of the values on $N$ more samples should be used).

\begin{figure}[t]
  \centering
  {\includegraphics[scale=0.4]{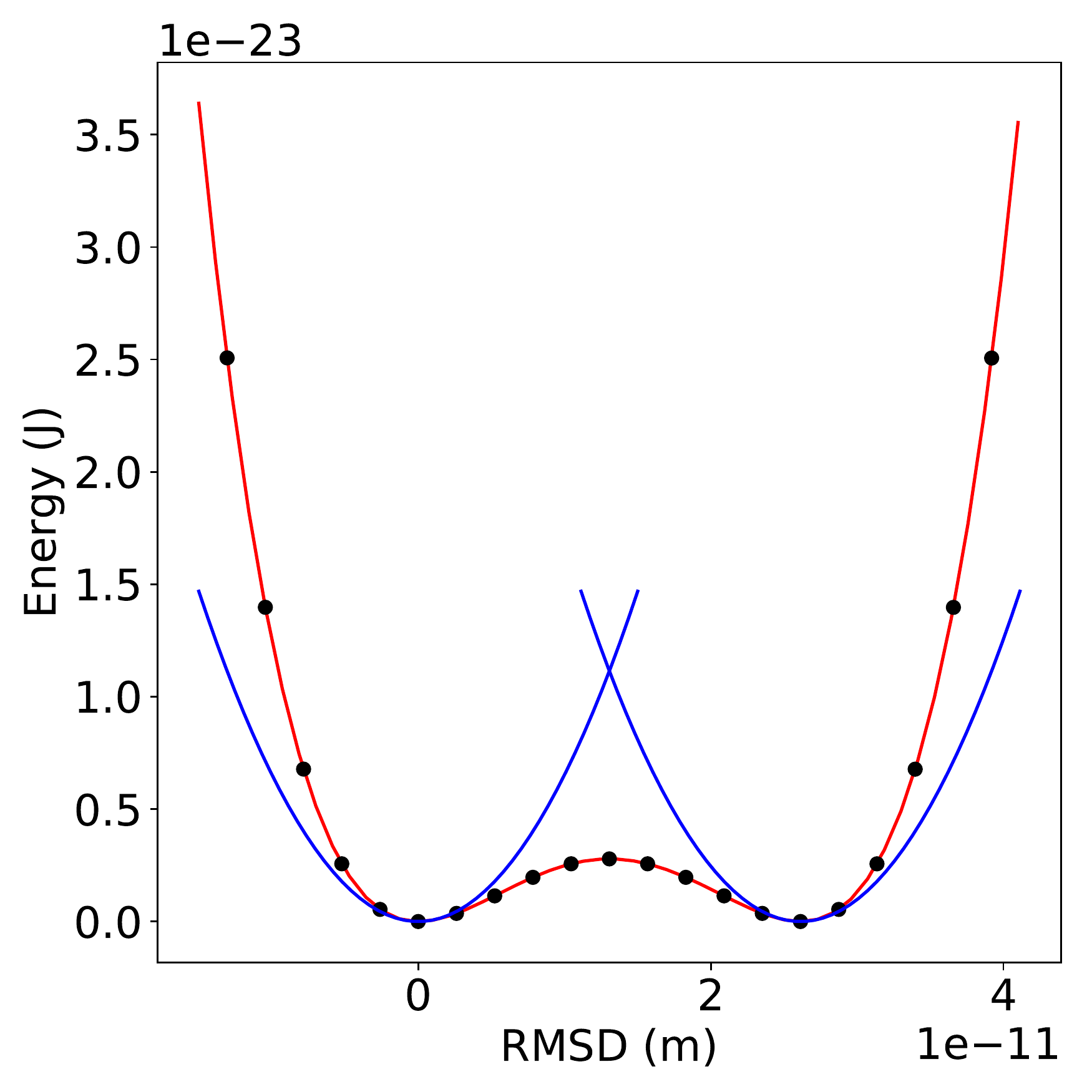}}
  \caption{ \footnotesize The double well potential found by linear interpolation
  between the two realizations for $N=300$. The black circles are found by linear interpolation, the red line is a fourth-order polynomial fit to the data. The blue curves show the harmonic approximations for two minima.}
   \label{fig:dws}
\end{figure}

$T_{\text{max}}$, in Table~\ref{tab:dwpvals},  denotes the temperature at which the specific heat of a two-level system is maximum. To find this temperature, we solved the Schrodinger's equation numerically using the embedding method (\cite{jelic2012double} and Appendix C in~\cite{sadjadidiss}) for the double-well potential in Fig.~\ref{fig:dws} and found its specific heat using energy levels. In general, $T_{\text{max}}$ happens to be at about $\sim10$ K which is much higher than
the range of temperatures at which the tunneling states are assumed to be active.

\begin{table}[t]
\centering
\caption{Characteristics of double-well potentials in SI units for four different system sizes, $N$.}
\label{tab:dwpvals}
\begin{center}
  \begin{tabular}{ c  c  c  c  c  c }
    \hline
    $N$ & $V_b$(J) & $T_b$(K) & $d$(\AA) & $\lambda$ & $T_{\text{max}} (K)$\\ \hline
    48 & $9.45 \times 10^{-27}$ & $6.84 \times 10^{-4}$ & $0.04$ & $0.0012$  & $16$ \\
    300 & $2.79\times 10^{-24}$ & $2 \times 10^{-1}$ & $0.30$ & $0.10$ & $12$ \\
    1254 & $1.44\times 10^{-25}$ & $1.04 \times 10^{-2}$ & $0.17$ & $0.02$ & $5$ \\
    5016 & $3.84\times 10^{-27}$ & $2.78\times 10^{-4}$ & $0.12$ & $0.0025$ & $23$ \\\hline
  \end{tabular}
\end{center}
\end{table}

We also repeated the same calculations in 3D glasses but very similar results were obtained. However, as it was discussed earlier in this section, packing of granular materials is another network matter which can be studied using this framework. In addition, because these networks are essentially different from networks glasses in terms of ring distribution, preparation method, etc they might
shed light on the nature of multiple realizations and possibly tunneling states from a different perspective.

Fig.~\ref{fig:jammed} shows an example of a computer-generated jammed packing with 247 vertices in which two realizations (original conformation and alternative realization found by the modified single-cut algorithm) are superimposed. As it is shown, in some regions the displacements are more pronounced. However, compared to network glasses, the mean displacement is about one order of magnitude smaller.

The jammed packings are generated by minimizing the overlap among circles in contact. In the original network, all overlaps are smaller than a prespecified threshold. To find the alternative realization, however, the interaction between soft disks are replaced by harmonic springs. Therefore, both realizations have zero energy if interactions are harmonic springs since all edge lengths are equal but since adjacent vertices can also move, disks that were previously non-overlapping can intersect leading to additional energy or the system can undergo an unjamming process. Hence, an important question is whether the alternative realization is at energy minima if the interactions are based on the overlapping soft disks. In our tests, we observed that second realizations of jammed systems have generally large overlaps between disks but we have observed some examples in which even alternative realizations are jammed and at the energy minimum or very close to it.

\begin{figure}[t]
  \centering
  {\includegraphics[scale=0.175]{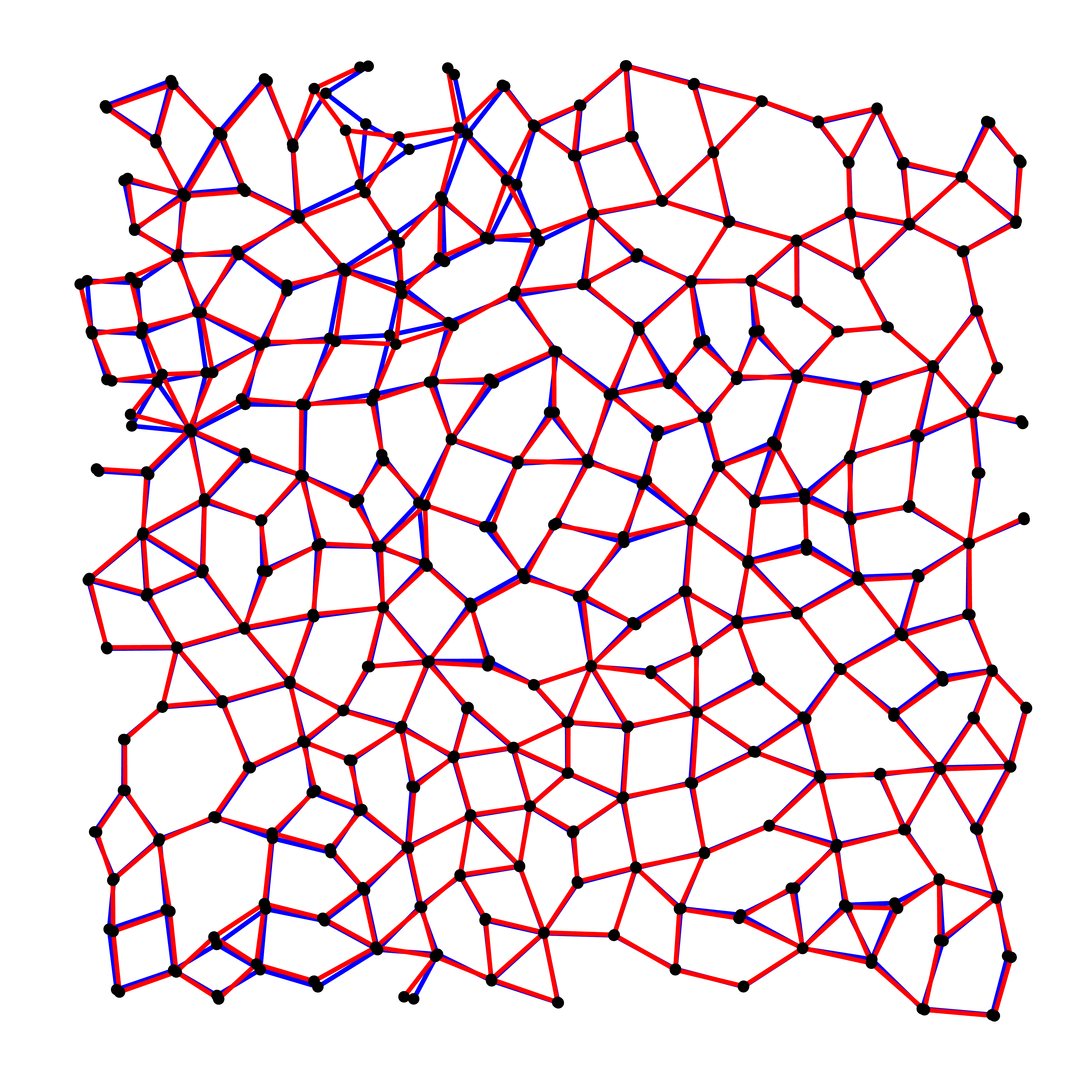}}
  \caption{\footnotesize  A jammed circle packing. In this representation, circles are replaced by their center. The original network is drawn with blue lines while the equivalent configuration found by modified single-cut algorithm is shown in red.}
    \label{fig:jammed}
\end{figure}

To draw a comparison between Trihex and 2D glasses, it seems that the
single-cut algorithm can only find solutions that belong to the same connected component
while realizations on other  connected components are energetically inaccessible since they
contain motions of larger units such as a rigid triangle (or tetrahedron).
Although, we think such branches exist in glasses, it is not computationally
feasible to find all the  connected components for such large systems similar to Trihex.

It is worth noting that the above discussion can be directly applied to bulk
glasses. We repeated the modified single edge cut algorithm for various silica
structures in three dimensions. The only modification required in 3D is that three
edges need to be removed to reach the isostatic point. Our results for bulk (3D) glasses
were very similar to the two-dimensional case.

\section{Questions}
There are a number of important open questions that we list here that require this work to be examined in a larger context.  The questions are general and go beyond the models considered in this paper.
\begin{itemize}
 \item The  maximum number of distinct solutions for the Trihex is 112.
 Note Mathematica often gives the same solution multiple times.  Where does
 this number come from? CayMos theory gives an upper bound of 128. Using numerology it is $2^7 -2^4$. Note that the problem cannot be reduced to a single polynomial where the number of real solutions is always less than or equal to the degree of the polynomial.
 \item Why nothing more complex than an open or closed retrograde trajectory found in all the examples here?
 \item Why are the trajectories all ``smooth'' with no singularities? This can, however, be proved in the generic case for the closed curves obtained in single-cut and CayMos algorithms, which are configuration spaces of 1-degree of freedom (dof) mechanisms.
 \item How general is this scheme?
 \item Why are there no bifurcations in the solutions?
\end{itemize}

\section{Discussion and Conclusion}
The main purpose of this collaboration is the present Theorems 1 and 2 and the single cut algorithm in a straightforward way so that it is accessible for future work. These are powerful statements about isostatic systems that we have just begun to explore here. It is hoped that the reader can see the potential and will pursue these methods further. We have taken the first steps with atomic clusters in three dimensions, and with tunnelling modes in glasses and jammed systems in two dimensions. We emphasize again that these approaches work in any dimension. The emphasis here on two dimensions is for simplicity and for ease of visualization. Theorem 1 is counter intuitive at first sight but becomes very natural when the circuit associated with a single cut (see Figure~\ref{fig:circuits}) is understood. We have shown that modestly large systems behave in the same way as smaller systems but the limit of large systems (tending to infinite size) are very different. This is important for solid state problems where systems have a size of the order of Avogadro's number ($\sim 10^{24}$) and this has important implications for the origin of tunnelling states in glasses. ``More is different''~\cite{anderson1972more}.

\section{Acknowledgment}
\begin{acknowledgments}
Finally one of us (MFT) would like to congratulate David Drabold upon reaching the temporal landmark of three score years and for all the interesting discussions at home and in more exotic locations that we have enjoyed over the years. His general approach to science and in particular his work on glassy networks has influenced the work in this paper.

The authors would like to thank financial support though NSF Grants
No. DMS 1564468 (Connelly, Gortler, Holmes-Cerfon, Sitharam, Thorpe). We would particularly like to thank Eric Corwin and Kenneth Stephenson for stimulating discussions.
\end{acknowledgments}

\bibliographystyle{apsrev4-2}
\bibliography{references,refs,library,dis}
\end{document}